\documentclass{article}
\usepackage[utf8]{inputenc}
\usepackage{amsfonts}
\usepackage{amsmath}
\usepackage{amsthm}
\usepackage{amssymb}
\usepackage{physics}
\usepackage{mathtools}
\usepackage{bbold}
\usepackage{graphicx}
\usepackage{xcolor}
\usepackage{float}
\usepackage{tcolorbox}
\usepackage{authblk}
\usepackage{cite}
\usepackage{hyperref}
\usepackage{etoolbox}
\newtheorem{theorem}{Theorem}
\usepackage[a4paper, total={6in, 9.5in}]{geometry}
\definecolor{caribbeangreen}{rgb}{0.0, 0.8, 0.6}

\theoremstyle{definition}

\title{Haar averaged moments of correlation functions and OTOCs in Floquet systems}
\author{Ewan McCulloch\footnote{email: ewan.r.mcculloch@gmail.com}}
\affil{\textit{University of Birmingham, Birmingham B15 2TT, UK}}
\date{\vspace{-5ex}}

\makeatletter
\let\newtitle\@title
\let\newauthor\@author
\let\newdate\@date
\makeatother

\begin{document}

\maketitle

\begin{abstract}
Scrambling and thermalisation are topics of intense study in both condensed matter and high energy physics. Random unitary dynamics form a simple testing-ground for our theoretical understanding of these processes. In this work, we derive exact expressions for the large $q$ limiting behaviour of a selection of $n$-point correlation functions, out-of-time-ordered correlators (OTOCs), and their moments. In the process we find a general principle that breaks OTOCs into small and easy to calculate pieces, and which can likely be deployed in a more general context.
\end{abstract}

\section{Introduction}
The thermalisation of quantum chaotic systems has brought into focus the concept of scrambling, where much of the locally accessible information in the initial state is sequestered into more complicated, and increasingly difficult to measure observables. Scrambling has been the topic of huge interest in the fields of holography and black holes \cite{Sekino08, Lashkari2013, Shenker2014a, Shenker2014b, Shenker2015, Maldacena2016, Hartman2013, Liu14a, Liu14b, Mezei16, Blake16}, quantum field theories \cite{Stanford2016, Asplund15, Banerjee2017, Roberts18, Roberts16, Swingle17, Aleiner16, CalabreseCardy05} and random unitary circuits \cite{Nahum16, Nahum17, RvK17, OTOCDiff1, OTOCDiff2, Brown12, ChanDeLuca1}, to name a few. 

Random unitary matrices have found many uses in the study of scrambling (and thermalisation more generally), having a crucial role in random unitary circuits models and also in a model of scrambling in black holes \cite{Hayden07}. Evolution in strongly ergodic systems looks approximately Haar random at intermediary times (the ‘dip’ in the spectral form factor\cite{Cotler2017a,Gharibyan2018}) before the onset of random matrix theory behaviour and the famous ‘ramp’ \cite{Cotler2017a,Meh2004,Brezin1997,Gharibyan2018}. It is no surprise then, that Haar random matrices are often used as a benchmark when comparing the scrambling dynamics of other models.

Previous work by \cite{RobertsDesign} investigates OTO correlators as a scrambling diagnostic and calculate the 4, 6 and 8-point OTO correlators (of the form $\langle A_1 B_1(1) \cdots A_k B_k(1)\rangle$ where $B(1)=UBU^\dagger$) for the Haar unitary ensemble. In this paper, we investigate the $n$-point correlators of the form $\langle Z(t_1)\cdots Z(t_n)\rangle$ where $Z\in \mathbb{C}^{q\times q}$ is a normalised, traceless operator and the discrete time evolution is generated by a Haar random unitary $U$. In section \ref{higher moments} we find the scaling behaviour for the Haar average of products of $p$ correlators to be $\order{q^{-2\lfloor p/2 \rfloor}}$ as $q\to\infty$. In section \ref{two-correlators} we find that for $p=2$, this scaling bound is met only when the correlators are complex conjugates of one another (Eq. \ref{correlators}). Additionally, the proportionality constant is a symmetry factor that counts the number of cyclic symmetries the correlators have. Finally in section \ref{OTOC section}, we investigate a class of OTOCs using a diagrammatic scheme with which we can systematically count all leading order contributions. The re-summation of these diagrams reveals a simple structure, in which the OTOC is divided into easily evaluated pieces that group together operators with a shared time evolution (Eq. \ref{OTOCs}).

These results are useful outside of the setting of a quantum dot; in an upcoming pre-print \cite{MMFpaper}, we study the hydrodynamic transport of quantum information in a one-dimensional Floquet model with on-site Haar random scrambling. OTOCs and other $n$-point functions appear naturally in this calculation; the results in this paper enable us to find $\order{1/q^2}$ corrections to the circuit averaged butterfly velocity. Using the methodology developed in section \ref{OTOC section}, one might be able to generalise the results of this paper to $n$-point functions of the more general form $\langle Z_1(t_1)\cdots Z_n(t_n)\rangle$, where each $Z_i\in \mathbb{C}^{q\times q}$ can be any normalised, traceless operator. This would open the door to circuit averaging in a much large number of Floquet circuits.

\section{Correlators and contours}
In this manuscript, an $n$-point correlation function is given by 
\begin{equation}
    \langle \mathcal{Z}(\boldsymbol{t})\rangle = \langle Z(t_1) \cdots Z(t_n)\rangle = \frac{1}{q}\Tr[ Z(t_1) \cdots Z(t_n) ],
\end{equation}
where $\boldsymbol{t}=\left(t_1,\cdots,t_N\right)$ and $\mathcal{Z}(\boldsymbol{t})=Z(t_1)\cdots Z(t_n)$ is a product of $n$ `scrambled' Pauli $Z$ matrices, $Z(t)=U^t Z U^{-t}$ (with a unitary $U$ drawn from the Haar distribution on the group $\mathcal{U}(q)$ of $q\times q$ unitary matrices), and where no consecutive times are equal (including the first and final times, which are consecutive due to the cyclic property of the trace). Two correlators with times $\boldsymbol{t}$ and $\boldsymbol{t}'$ are identical if the two sequences $\boldsymbol{t}$ and $\boldsymbol{t}'$ equal up to a cyclic permutation. Therefore, without loss of generality, we assume that $t_1\leq t_i$ for all $i$.

A correlator $\langle \mathcal{Z}(\boldsymbol{t})\rangle$ is called contour-ordered for a contour $\mathcal{C}$ if the sequence of times $\boldsymbol{t}=(t_1,\cdots,t_n)$, is contour ordered on $\mathcal{C}$. An example is given below for contours with only a single forward and backward segment and for contours with two forward and backward segments. The later is the type of contour that out-of-time-ordered (OTO) correlators live on, and hence is referred to as a contour of OTO type. The former is referred to as a contour of time-ordered (TO) type.

\begin{figure}[H]
    \centering
    \raisebox{-0.45\totalheight}{\includegraphics[height = 2.1cm]{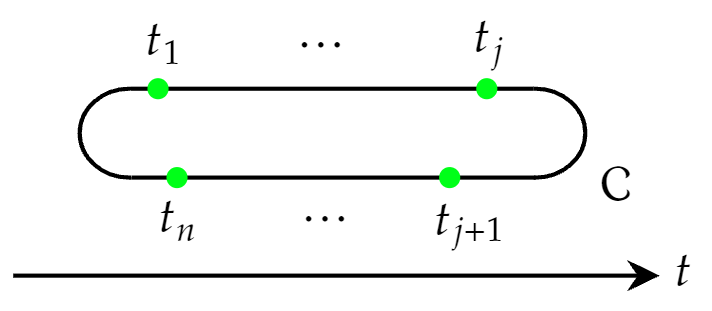}}
    \hspace{5mm}
    \raisebox{-0.3\totalheight}{\includegraphics[height = 3.1cm]{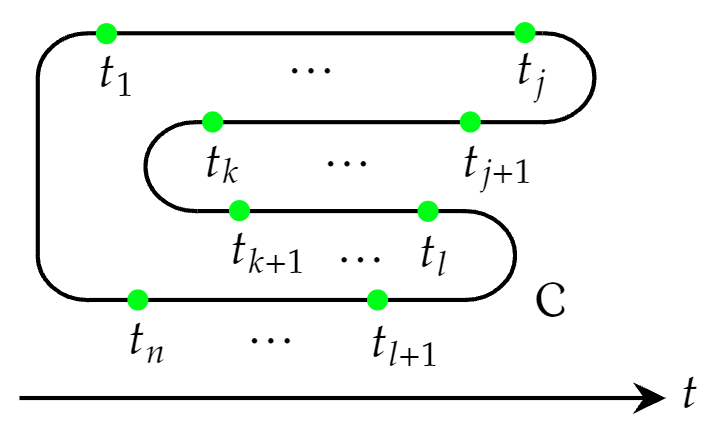}}
    \caption{Contour ordering for TO type contours and OTO type contours.}
    \label{contours}
\end{figure}

Correlators that are contour ordered for a TO (OTO) type contour are referred to as TO (OTO) correlators. In the following, we assume that the number of instances $\mathcal{N}_U$ of the Haar random unitary $U$ appearing in an expression to be averaged does not scale with the dimension $q$ and find large $q$ asymptotics. For TO correlators, this simply means that $(\max(t_i)-\min(t_i))$ is finite and fixed while we probe large $q$. For OTO correlators, the number of instance of $U$ is bounded above by $2(\max(t_i)-\min(t_i))$.

\section[\texorpdfstring{tex}{pdfbookmark}]{$n$-point correlatation functions}\label{correlator}
In this section we investigate the large $q$ scaling behaviour (of the Haar average) of arbitrary-time-ordered correlators $\langle \mathcal{Z}(\boldsymbol{t})\rangle$, where the time-ordering is arbitrary as we make no assumptions about the sequence $\boldsymbol{t}$ beyond the condition that no consecutive times are equal.
\begin{theorem}
	The Haar averaged of an $n$-point correlation function $\langle \mathcal{Z}(\boldsymbol{t})\rangle$ is $\order{1/q^2}$,
	\begin{equation}
	\int dU \langle \mathcal{Z}(\boldsymbol{t}) \rangle = \order{1/q^2} \quad \textrm{as $q\to\infty$}.    
	\end{equation}
\end{theorem}
\begin{proof}
Label the differences $x_i = t_{i+1} - t_i \neq 0$ with $x_1=t_1-t_n$, such that $\sum_i x_i = 0$. The number of instances of $U$ appearing in this correlators is given by $\mathcal{N}_U=\sum_{i}\abs{x_i}/2$.
This allows us to write the $n$-point function as
\begin{equation}
\langle Z(t_1) \cdots Z(t_N) \rangle = \langle ZU^{x_1}ZU^{x_2} \cdots ZU^{x_n} \rangle
\end{equation}
Taking a Haar average over the unitary $U$ and making use of the left and right invariance of the Haar measure $\int dU f(U) = \int d(UV) f(UV) = \int dU f(UV)$ (similarly for left invariance) and by representing the moments of the Haar measure as a sum over permutations weighted by the Weingarten functions \cite{Collins2002MomentsAC,Collins2006,asymptoticweingarten}, we have
\begin{align*}
\int dU \langle \mathcal{Z}(\boldsymbol{t})\rangle &= \iint dU dV \langle Z V U^{x_1}V^{\dagger}ZVU^{x_2} \cdots V U^{x_n}V^{\dagger} \rangle \\
&= \int dU \sum_{\sigma,\tau\in S_n}\frac{\text{Wg}(q,\tau\sigma^{-1})}{q} \times \quad \raisebox{-0.45\totalheight}{\includegraphics[height=2.1cm]{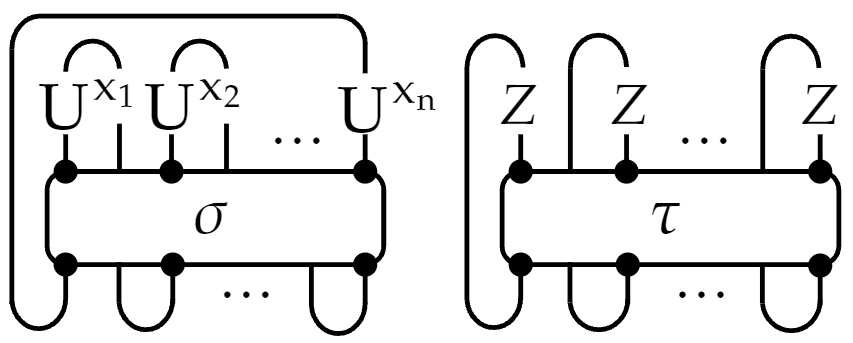}}
\end{align*}
Where we have used the Weingarten formula for the unitary group of dimension $q$ \cite{Collins2002MomentsAC,Collins2006} to evaluate the Haar integral in $V$, $\text{Wg}(q,\sigma)$ is the Weingarten function \cite{asymptoticweingarten}. We have used the following convention for the legs of tensors,
\begin{equation}
	\centering
	\includegraphics[height = 2cm]{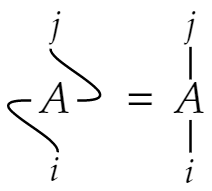}.
\end{equation}
Define the permutation $\pi=(n,1,2,\cdots,n-1)$ (in cycle notation). Diagrammatically, with the lower legs as the incoming legs, this is given by
\begin{equation}
	\pi=\vcenter{\hbox{\includegraphics[height = 1cm]{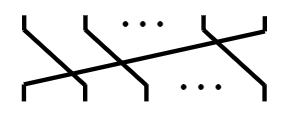}}}
\end{equation}
By inserting the identity permutation as $\mathbb{1}=\pi \pi^{-1}$, we can simplify the contraction of the $U^{x_i}$ using the following,
\begin{equation}
	\vcenter{\hbox{\includegraphics[height = 2.5cm]{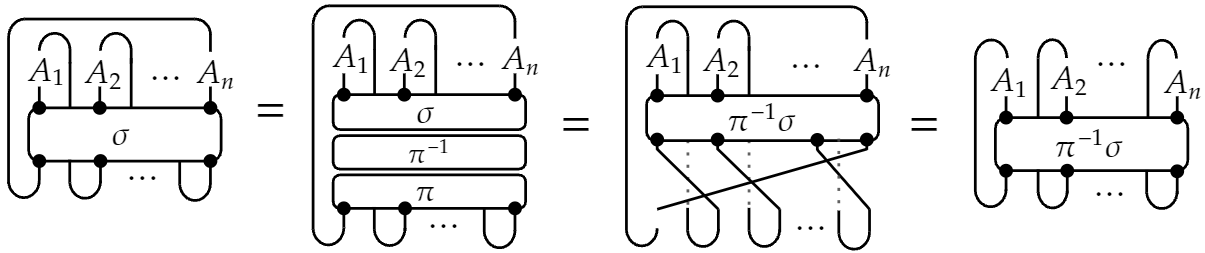}}}.
\end{equation}
Allowing the change of variables in the sum $\sigma \to \pi \sigma$ gives
\begin{equation}
\int dU \langle \mathcal{Z}(\boldsymbol{t})\rangle = \sum_{\sigma,\tau\in S_n}\frac{\text{Wg}(q,\tau\sigma^{-1}\pi)}{q} G(\boldsymbol{x},\sigma)H(\tau),
\end{equation}
where $H(\tau)$ is given by the diagram
\begin{equation}\label{Hdef}
	\vcenter{\hbox{\includegraphics[height = 2.1cm]{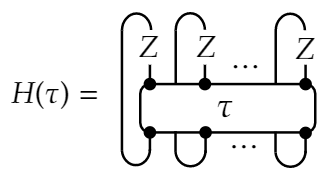}}}.
\end{equation}
The evaluation of $H(\tau)$ is simple,
\begin{equation}
H(\tau)= 
\begin{cases}
q^{\abs{C_\tau}},& \text{if } \tau \text{ has cycles only of even length}\\
0,              & \text{otherwise}.
\end{cases}
\end{equation}
An immediate consequence of this is that for odd $n$, $H(\tau)=0$ for all $\tau$. From now on we consider only even $n$. let $\mathcal{H}\subset S_n$ refer to the set of permutations with cycles of even length only, we can restrict the sum over $\tau\in S_n$ to $\tau\in \mathcal{H}$ in the following. 
The function $G(\boldsymbol{x},\sigma)$ is given by Haar averaged diagram below
\begin{equation}
	G(\boldsymbol{x},\sigma) = \int dU \vcenter{\hbox{\includegraphics[height = 2.1cm]{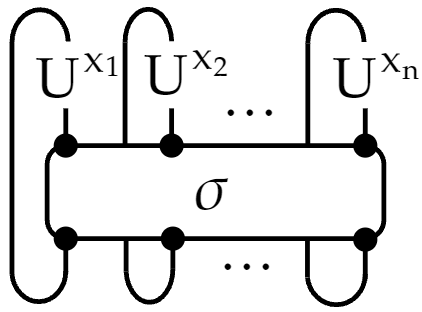}}}
\end{equation}
For a given permutation $\sigma$ and time differences $\boldsymbol{x}$, $G(\boldsymbol{x},\sigma)$ is given by the Haar average of a product of loops, each carrying a power of $U$.
\begin{equation}
G(\boldsymbol{x},\sigma) = q^{N_\text{ud}}\int dU \prod^l_{m=1}\Tr(U^m)^{a_m}\Tr(U^{-m})^{b_m}
\end{equation}
where the $a_m$ and $b_m$ count the multiplicity of the decorated loop corresponding with $\Tr(U^m)$ and $\Tr(U^{-m})$ respectively and depend on $\boldsymbol{x}$ and the permutation $\sigma$. $l$ is the maximum power of $U$ or $U^{-1}$ decorating a loop in the contraction, whichever is larger. $N_\text{ud}$ is the number of undecorated loops (loops that carry no factors of $U$ or $U^{-1}$) and also depends on $\sigma$ and $\boldsymbol{x}$. To evaluate this Haar average we use a result of \cite{Diaconis_2001}.

For $q \geq \max(\sum_{m=1}^l ma_m, \sum_{m=1}^l mb_m)$
\begin{equation}
\int dU \prod^l_{m=1}\Tr(U^m)^{a_m}\Tr(U^{-m})^{b_m} = \delta^{\boldsymbol{a},\boldsymbol{b}}\prod^l_{m=1} m^{a_m}a_m!\leq \delta^{\boldsymbol{a},\boldsymbol{b}} N!, \textrm{ where } N = \sum_{m=1}^l ma_m \leq \mathcal{N}_U,
\end{equation}
where $\boldsymbol{a}=(a_1,\cdots,a_l)$ and $\boldsymbol{b}=(b_1,\cdots,b_l)$. This result is useful whenever one of the tensor diagrams above has at least two decorated loops (a single decorated loop is impossible as $\sum_i x_i = 0$). Importantly, the result does not scale with $q$ as $q\to\infty$, which allows us to write
\begin{equation}
G(\boldsymbol{x},\sigma) \leq C(\boldsymbol{x}) q^{N_\text{ud}}, \quad \text{for some } C(\boldsymbol{x}) > 0 \textrm{ independent of $q$},
\end{equation}
where $N_{\textrm{ud}}\leq n/2$, with equality only achieved when $\sigma$ is composed of $n/2$ disjoint transpositions, producing $n/2$ loops on which the $U^{x_i}$ pair up in such a way that every $U^{x_i}$ pairs with $U^{-x_i}$. This pairing is only possible for select $\boldsymbol{x}$. Call $\mathcal{G}\in S_n$ the set of permutations that produce only undecorated loops. The value of such a loop contraction is given by $q^{\abs{C_\sigma}}$, where $\abs{C_\sigma}$ is the number of cycles in $\sigma$. The Haar average is trivial, giving
\begin{equation}
    G(\boldsymbol{x},\sigma) = q^{\abs{C_\sigma}}, \textrm{ for $\sigma \in \mathcal{G}$.}
\end{equation}
The Haar averaged $n$-point function is now given by
\begin{align}\label{pre approx}
\begin{split}
\int dU \langle \mathcal{Z}(\boldsymbol{t})\rangle =& \sum_{\sigma\in\mathcal{G},\tau\in \mathcal{H}}\frac{\text{Wg}(q,\tau\sigma^{-1}\pi)}{q}q^{\abs{C_\sigma} + \abs{C_\tau}}\\
&+ \sum_{\sigma\in\overline{\mathcal{G}},\tau\in \mathcal{H}}\frac{\text{Wg}(q,\tau\sigma^{-1}\pi)}{q}q^{\abs{C_\tau}}G(\boldsymbol{x},\sigma)
\end{split}
\end{align}
where $\overline{\mathcal{G}}$ is the set of all the elements in $S_n$ not in $\mathcal{G}$. We use the large $q$ asymptotic form of the Weingarten function \cite{asymptoticweingarten,Collins2006},
\begin{equation}\label{asymptotic_weingarten}
\text{Wg}(q,\sigma) = \frac{1}{q^{n+\abs{\sigma}}}\prod_{c\in C_\sigma} (-1)^{\abs{c}-1}\text{Cat}_{\abs{c}-1} + \order{\frac{1}{q^{n+\abs{\sigma}+2}}}
\end{equation}
where $\abs{\sigma}$ is the minimum number of transposition that $\sigma$ is a product of, $C_\sigma$ is the set of cycles in $\sigma$ and $\abs{c}$ is the length of a cycle $c\in C_\sigma$. $\text{Cat}_i$ are the Catalan numbers. The terms in the second sum in Eq. \ref{pre approx} are all of the size $\order{q^{\abs{C_\tau}+N_\text{ud}-n-\abs{\tau\sigma^{-1}\pi} - 1}}$. Using $\abs{C_\tau} \leq n/2$ and $N_\text{ud}\leq (n-2)/2$, it is not hard to see that all these terms are at most $\order{1/q^2}$. Moreover, the number of these terms scales with $n$ and not $q$. This gives,
\begin{equation}\label{post approx}
	\int dU \langle \mathcal{Z}(\boldsymbol{t})\rangle = \sum_{\sigma\in\mathcal{G},\tau\in \mathcal{H}}a(\tau\sigma^{-1}\pi)q^{r(\sigma,\tau) - 1} + \order{1/q^2}
\end{equation}
where $a(\tau\sigma^{-1}\pi) = \prod_{c\in C_\sigma} (-1)^{\abs{c}-1}\text{Cat}_{\abs{c}-1}$ is independent of $q$ and $r(\sigma,\tau) = \abs{C_\tau}+\abs{C_\sigma}-n-\abs{\tau\sigma^{-1}\pi}$. Because $\sum_i x_i = 0$, each $U^{x_i}$ must be accompanied by at least one other $U^{x_j}$ in order to form an undecorated loop. Therefore the maximum number of cycles in $\sigma$ is $n/2$, when all $U^{x_i}$ are paired. This gives $\abs{C_\tau}+\abs{C_\sigma}-n \leq 0$.

Assuming $\abs{C_\tau}+\abs{C_\sigma}$ is saturated then both $\tau$ and $\sigma$ are composed of $n/2$ disjoint transpositions. Therefore they must have the same parity, $P(\sigma)=P(\tau)$. Assume also that $\abs{\tau\sigma^{-1}\pi}$ is minimised, $\abs{\tau\sigma^{-1}\pi}=0$. This implies that the $\sigma\tau^{-1}=\pi$. Applying the parity operator, we find $P(\sigma\tau^{-1})=P(\sigma)P(\tau)=P(\pi)$. Using $P(\sigma)=P(\tau)$ we find $P(\pi)=1$. However, $\pi$ is composed of a product of $n-1$ adjacent transpositions, for even $n$ this gives $P(\pi)=-1$. This is a contradiction, and implies that $\abs{C_\tau}+\abs{C_\sigma}$ and $\abs{\tau\sigma^{-1}\pi}$ cannot be simultaneously maximised and minimised respectively. Therefore $r(\sigma,\tau)\leq -1$. Every term in the \ref{post approx} is of size $\order{q^{-2}}$ or smaller,
\begin{equation}
\int dU \langle \mathcal{Z}(\boldsymbol{t})\rangle = \order{1/q^2}\quad \textrm{as } q\to\infty
\end{equation}
\end{proof}
\section{Higher moments}\label{higher moments}
In this section we generalise the analysis in section \ref{correlator} and determine bounds on the scaling behaviour (of Haar average) of the products of $p$ correlators.
\begin{theorem}\label{productofcorrelators}
	The Haar average of a product of $p$ correlators $\langle \mathcal{Z}(\boldsymbol{t}^1)\rangle \cdots \langle \mathcal{Z}(\boldsymbol{t}^p)\rangle $, with differing consecutive times $t^{(i)}_j\neq t^{(i)}_{j-1}, \ t^{(i)}_j \in \mathbb{Z}$, has the scaling behaviour
	\begin{equation}
	    \int dU \langle \mathcal{Z}(\boldsymbol{t}^1)\rangle\cdots \langle\mathcal{Z}(\boldsymbol{t}^p)\rangle = \order{1/q^{2\lfloor p/2 \rfloor}} \quad \textrm{as $q\to\infty$}.
	\end{equation}
\end{theorem}

\begin{proof}
The Haar average of a product of correlation functions, $\langle \mathcal{Z}(\boldsymbol{t}^1)\rangle\cdots \langle\mathcal{Z}(\boldsymbol{t}^p)\rangle$, is written
\begin{equation}
\int dU \langle \mathcal{Z}(\boldsymbol{t}^1)\rangle\cdots \langle\mathcal{Z}(\boldsymbol{t}^p)\rangle = \int dU \prod^p_{i=1} \langle Z(t^{(i)}_1) \cdots Z(t^{(i)}_{n_i}) \rangle.
\end{equation}
As before, define the differences $x^{(i)}_j = t^{(i)}_{j+1} - t^{(i)}_j$ with $x^{(i)}_{n_i} = t^{(i)}_1-t^{(i)}_{n_i}$ such that $x^{(i)}_j\neq 0$ and $\sum_jx^{(i)}_j=0$. The number of instances of $U$ appearing in the product of $p$ correlators is given by $\mathcal{N}_U=\sum_{i,j}\abs{x^{(i)}_j}/2$. Using the same trick as in section \ref{correlator} (representing the Haar integral over an auxiliary unitary as a weighted sum over permutations), we find
\begin{equation}
\int dU \langle \mathcal{Z}(\boldsymbol{t}^1)\rangle\cdots \langle\mathcal{Z}(\boldsymbol{t}^p)\rangle = \sum_{\sigma,\tau\in S_M} \frac{Wg(q,\tau\sigma^{-1})}{q^p} H(\tau)\int dU \raisebox{-0.45\totalheight}{\includegraphics[height=2.1cm]{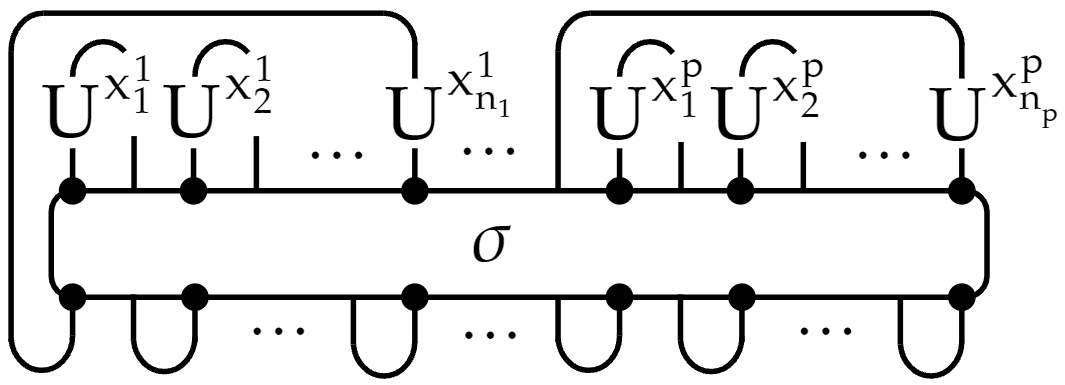}}
\end{equation}
where $M = \sum_i n_i$ and the trace diagram $H(\tau)$ is as in Eq. \ref{Hdef}, but with $M$ legs in total. We define the permutation $\pi = \pi_1 \pi_2 \cdots \pi_p$ where $\pi_i=(M_{i-1} + n_i,M_{i-1} + 1,M_{i-1} + 2,\cdots,M_{i-1}+n_i-1)$ where $M_k=\sum_{i=1}^k n_k$,
\begin{equation*}
\raisebox{-0.45\totalheight}{\includegraphics[height=1.2cm]{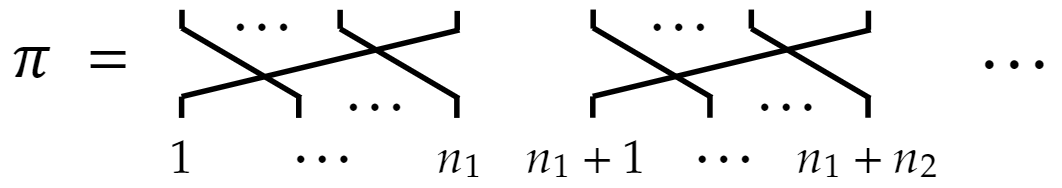}}.
\end{equation*}
For the set of indices $\{M_{i-1}+1,M_{i-1}+2,\cdots,M_{i-1}+n_i\}$ associated with the $i$-th of the $p$ correlators, $\pi_i$ does the same job that the single $\pi$ did in the previous proof. As before we shift the sum variable $\sigma\to\pi\sigma$ to arrive at
\begin{equation}
\int dU \langle \mathcal{Z}(\boldsymbol{t}^1)\rangle\cdots \langle\mathcal{Z}(\boldsymbol{t}^p)\rangle = \sum_{\sigma,\tau\in S_M} \frac{Wg(q,\tau\sigma^{-1}\pi)}{q^p} G(\{\boldsymbol{x}^i\},\sigma)H(\tau)
\end{equation}
where $G(\{\boldsymbol{x}^i\},\sigma)$ is defined as in the previous proof but the variable indexing order is given by $(x^{(1)}_1,\cdots, x^{(1)}_{N_1},x^{(2)}_1,\cdots,x^{(p)}_{N_p})$. As before, we split the sum into two pieces, the first where  $G(\{\boldsymbol{x}^i\},\sigma)$ contains at least two decorated loops and the second, where $G$ contains no decorated loops. Each term in the first sum is at most $\order{q^{-p-1}}$. We go about bounding the terms in the second sum in the same manner as before, by attempting to maximise $G\leq q^{\abs{C_{\sigma}}}$ and $H\leq q^{\abs{C_\tau}}$ and minimising $\textrm{Wg}(q,\tau\sigma^{-1}\pi)$. However, there is now an additional subtlety, the parity of $\pi$ is now dependent on $p\pmod 2$. For odd $p$, precisely the same argument can be made as before and we find $r(\sigma,\tau)\leq -1$, whereas for even $p$, we must use the weaker bound $r(\sigma,\tau)\leq 0$. The terms in the sum over permutations are then individually of size at most $\order{q^{-2\lfloor p/2 \rfloor}}$. Moreover, the number of these terms depends only on the parameter $M$ and $\mathcal{N}_U$, the number of instances of $U$ appearing in the product of $p$ correlators, neither of these scale with $q$. Therefore, We find the desired result,
\begin{equation}
\int dU \langle \mathcal{Z}(\boldsymbol{t}^1)\rangle\cdots \langle\mathcal{Z}(\boldsymbol{t}^p)\rangle = \order{1/q^{2\lfloor p/2 \rfloor}} \quad \textrm{as } q\to\infty
\end{equation}
\end{proof}

\section{Product of two correlation functions}\label{two-correlators}
In the previous sections we have found the bounds on the scaling behaviour of Haar averaged correlators and products of correlators as $q\to\infty$. In this section we present results that also predict the proportionality constants for the Haar average of a product of two correlators.
\begin{theorem}\label{correlators}
	The Haar average of a product of ATO's $\langle \mathcal{Z}(\boldsymbol{t}) \rangle$ and $\langle \mathcal{Z}(\boldsymbol{t}') \rangle^*$, where $\boldsymbol{t}=(t_1,\cdots,t_n)$ and $\boldsymbol{t}'=(t'_1,\cdots,t_{n'}')$,
	is given by
    \begin{equation}
        \int dU \langle \mathcal{Z}(\boldsymbol{t}) \rangle \langle \mathcal{Z}(\boldsymbol{t}') \rangle^* = \frac{1}{q^2}\sum_{\textrm{cyclic perm } \alpha}\delta^{\boldsymbol{t}',\alpha(\boldsymbol{t})} + \order{1/q^{3}},
    \end{equation}
    where $\delta^{\boldsymbol{t}',\alpha(\boldsymbol{t})}=1$ if the sequences of times $\boldsymbol{t}'$ and $\alpha(\boldsymbol{t})$ are equal up to a global time translation and where $\alpha(\boldsymbol{t})$ is the cyclic permutation $\alpha$ of the sequence $\boldsymbol{t}$. The delta constraint is zero otherwise.
\end{theorem}
This result can be equivalently states as
\begin{equation}
    \int dU \langle \mathcal{Z}(\boldsymbol{t}) \rangle \langle \mathcal{Z}(\boldsymbol{t}') \rangle^* = \frac{S(\boldsymbol{t})}{q^2}\tilde{\delta}^{\boldsymbol{t}',\boldsymbol{t}} + \order{1/q^3},
\end{equation}
where $\tilde{\delta}^{\boldsymbol{t}',\boldsymbol{t}}$ checks whether the times $\boldsymbol{t}'$ and $\boldsymbol{t}$ are equal up to global time-translations and cyclic permutations and $S(\boldsymbol{t})$ is a symmetry factor that counts the number of cyclic symmetries $\boldsymbol{t}$ has. For instance, if $\langle \mathcal{Z}(\boldsymbol{t})\rangle$ is a TO correlator, $S(\boldsymbol{t})=1$, while for OTO correlators $1 \leq S(\boldsymbol{t})\leq 2$. For a contour with $m$ forward and backward segments, the symmetry factor obeys $1 \leq S(\boldsymbol{t})\leq m$.

\begin{proof}
An equivalent representation of the product of correlators is given by
\begin{equation}
\langle \mathcal{Z}(\boldsymbol{t}) \rangle \langle \mathcal{Z}(\boldsymbol{t}') \rangle^* = \langle \prod_i (ZU^{x_i}) \rangle \langle \prod_j (ZU^{y_j}) \rangle^*
\end{equation}
where $x_i= t_{i+1}-t_i$ and the edge case $x_1 = t_1 - t_n$. Similarly, $y_i= t'_{i+1}-t'_i$ and the edge case $y_1 = t'_{n'} - t_1$. By assumption, all consecutive times differ so that all $x_i\neq 0$ and $y_j\neq 0$. By using the left and right invariance of the Haar measure once again, we write the Haar averaged expression as
\begin{equation}\label{av_prod_corr_Weingarten}
    \int dU \langle \mathcal{Z}(\boldsymbol{t}) \rangle \langle \mathcal{Z}(\boldsymbol{t}') \rangle^* = \frac{1}{q^2}\sum_{\sigma,\tau\in S_N}Wg(\sigma\tau^{-1})\tilde{G}(\boldsymbol{x},\boldsymbol{y},\sigma)H(\tau),
\end{equation}
where $N\equiv n+n'$, $\boldsymbol{x}=(x_1,\cdots,x_n)$, $\boldsymbol{y}=(y_1,\cdots,y_{n'})$ and where $H(\tau)$ is as given in Eq. \ref{Hdef}, but with $N$ legs. As before, $H(\tau)$ is zero unless $N$ is even, in the remainder of this proof we will assume this is so. This proof differs from the previous proofs in that we choose not to shift the permutation $\sigma$ by some fixed permutation $\pi$, we instead have defined the function $\tilde{G}(\boldsymbol{x},\boldsymbol{y},\sigma)$, which differs in wiring from the function $G$ defined the previous sections. It is given by
\begin{equation}
    \tilde{G}(\boldsymbol{x},\boldsymbol{y},\sigma)=\int dU
	\raisebox{-0.45\totalheight}{\includegraphics[height = 2.1cm]{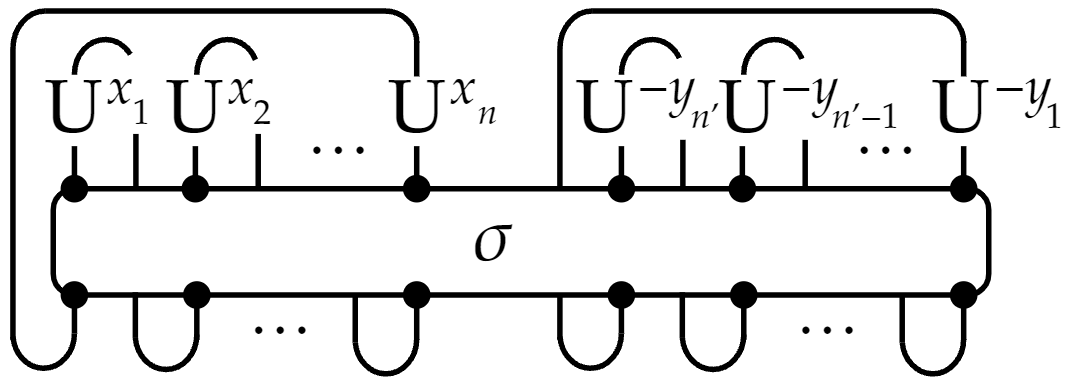}}.
\end{equation}
Despite differing from the $G$ of the previous proofs, the same argument of counting undecorated loops holds. When each $U^a$ decoration pairs up with a $U^{-a}$ on $N/2$ loops, $\tilde{G}(\boldsymbol{x},\boldsymbol{y},\sigma)$ is maximised (for large $q$), giving $\tilde{G}(\boldsymbol{x},\boldsymbol{y},\sigma)=q^{N/2}$. Whereas, $H(\tau)$ is maximised when all the $Z$'s are paired up on separate loops which gives a maximum of $N/2$ loops, giving $H(\tau) = q^{N/2}$. We therefore have two bounds
\begin{equation}
    0\leq H(\tau)\leq q^{N/2}, \ 0\leq \tilde{G}(\boldsymbol{x},\boldsymbol{y},\sigma)\leq q^{N/2}
\end{equation}
Using the asymptotic form of the Weingarten function in Eq. \ref{asymptotic_weingarten}, the sum in Eq. \ref{av_prod_corr_Weingarten} is bounded by $C(\boldsymbol{x},\boldsymbol{y}) q^{N/2+N/2-N-|\sigma\tau^{-1}|-2}=C(\boldsymbol{x},\boldsymbol{y}) q^{-|\sigma\tau^{-1}|-2}$, for some $C(\boldsymbol{x},\boldsymbol{y})$ that is independent of $q$. This is $\order{q^{-2}}$ only when $\sigma=\tau$ and $\leq\order{q^{-3}}$ otherwise. As our goal is to calculate the leading order $1/q^2$ contributions, the inequality above indicates that we need only consider cases where $\sigma=\tau$, and where $\tau$ is a product of $N/2$ disjoint transpositions. This gives:
\begin{equation}
    \int dU \langle \mathcal{Z}(\boldsymbol{t}) \rangle \langle \mathcal{Z}(\boldsymbol{t}') \rangle^* = \frac{1}{q^{2+N/2}}\sum_{\sigma\in \mathcal{T}_N}\tilde{G}(\boldsymbol{x},\boldsymbol{y},\sigma) + \order{q^{-3}}
\end{equation}
where $\mathcal{T}_N$ is the set of permutations $\sigma$ that are products of $N/2$ disjoint transpositions. Given such a permutation $\sigma$, suppose that it contains a transposition $T_{1,i+1}$ between legs $1$ and $i+1$ and that $i+1\leq n$ (such that the transposition is `within' the indices of first correlation function). Such a transposition puts the unitaries $U^{x_1}$ and $U^{x_i}$ on the same wire. We remind the reader that the theoretical maximum of $\tilde{G}$ is achieved if all unitaries are paired ($U^x$ with $U^{-x}$) on separate loops. Therefore, to ensure that no other unitaries share the loop with $U^{x_1}$ and $U^{x_i}$, we must choose the next transposition as $T_{2,i}$, this closes the wire with $U^{x_1}$ and $U^{x_i}$ into a loop. This is shown below,

\begin{figure}[H]
	\centering
	\includegraphics[height = 2.5cm]{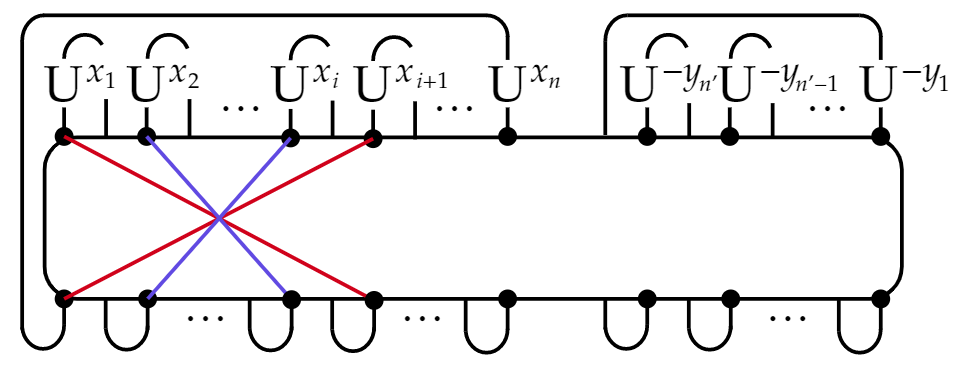}
\end{figure}
The requirement that unitaries are paired only, forces us pick transpositions that cascade inwards until one of two scenarios occur, depending on whether $i$ is even or odd. If $i$ is even, then the cascade ends as shown below, with a loop with a single unitary decorating it. This violates the conditions for maximising $\tilde{G}$.
\begin{figure}[H]
	\centering
	\includegraphics[height = 2.3cm]{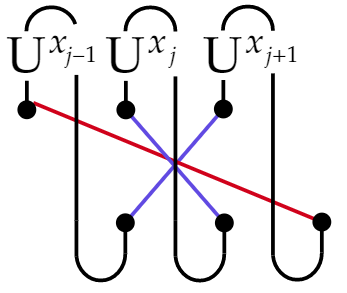}
\end{figure}
If $i$ is odd, then the cascade ends by forcing the choice of a $1$-cycle, instead of a transposition, in order to have the unitaries paired up. This permutation is not a member of $\mathcal{T}_N$, and so this possibility is dismissed.
\begin{figure}[H]
	\centering
	\includegraphics[height = 2.3cm]{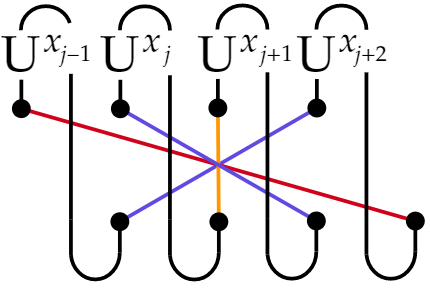}
\end{figure}
We need not have started with the transposition $T_{1,i+1}$ with $i<n$. If $\sigma$ contains any transpositions $T_{i,j}$ with $1\leq i,j\leq n$, $i\neq j$, then the same cascade argument above can be used to demonstrate that we cannot pair up unitaries onto separate loops with disjoint transpositions only. The same argument can be made by starting with a transposition exclusively within the second correlator. We therefore conclude that in order to maximise $\tilde{G}$, $\sigma$ must be comprised of disjoint transpositions where every transposition spans between the two correlators. An immediate consequence of this is that the number of legs in each correlator must be equal, $n=n'$. Assuming a transposition $T_{1,n+1+i}$, $0\leq i< n$, we again employ the cascade argument, where the transposition $T_{2,n+i}$ is forced next and then $T_{3,n+i-1}$ and so on until we stop at the transposition $T_{i+1,n+1}$. This is shown in the diagram below, where we have only displayed the time labels to declutter the picture.
\begin{figure}[H]
	\centering
	\includegraphics[height = 3cm]{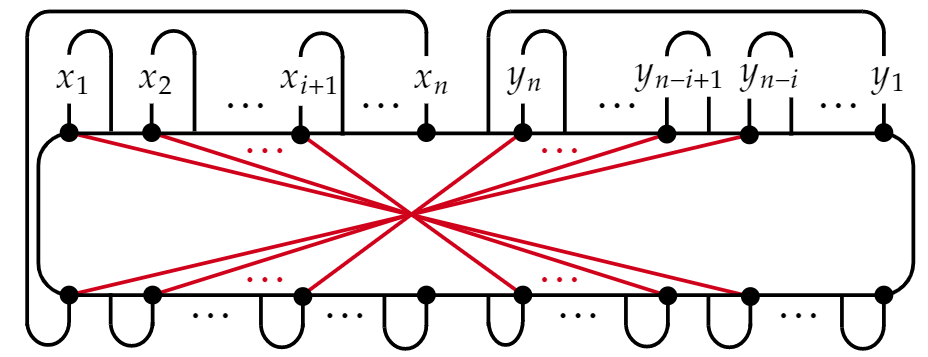}
\end{figure}
Notice that the final transposition added places the unitaries $U^{x_{i+1}}$ and $U^{-y_1}$ on the same wire, in order to close this loop with no other unitaries decorating it, we must choose the transposition $T_{i+2,2n}$, this starts another cascade with the next forced transposition being $T_{i+3,2n-1}$, until we finish with the final transposition $T_{n,n+2+i}$. This is shown below.
\begin{figure}[H]
	\centering
	\includegraphics[height = 3cm]{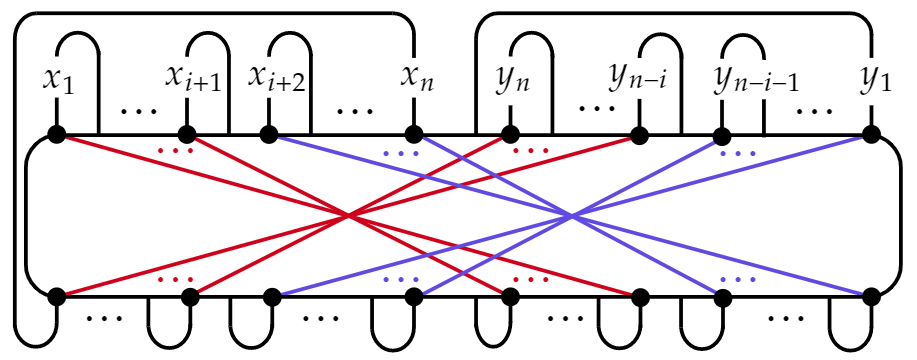}
\end{figure}
It is clearer to distinguish the two blocks of transpositions as shown below with the first block having $i$ transpositions and the second having $n-i$.
\begin{figure}[H]
	\centering
	\includegraphics[height = 3.2cm]{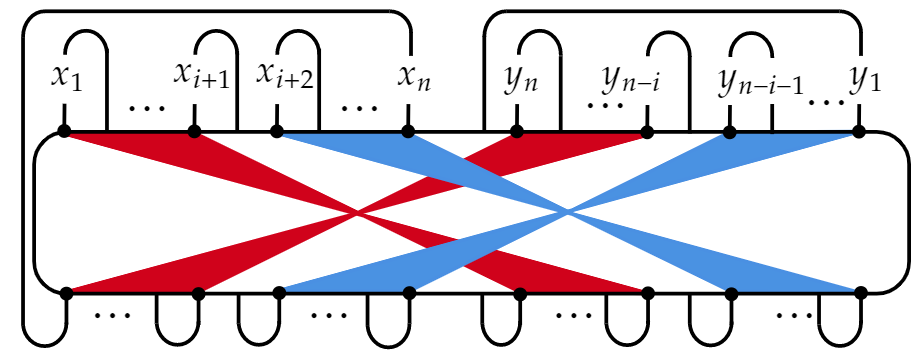}
	\caption{permutations $\sigma$ that pair the $U^{x_i}$ and $U^{-y_j}$ onto separate loops. \label{butterfly_cascade}}
\end{figure}
Therefore, only permutations of the form below may contribute at leading order ($\tilde{G}=q^{N/2}$).
\begin{equation}
    \sigma_m=
	\raisebox{-0.55\totalheight}{\includegraphics[height = 1.8cm]{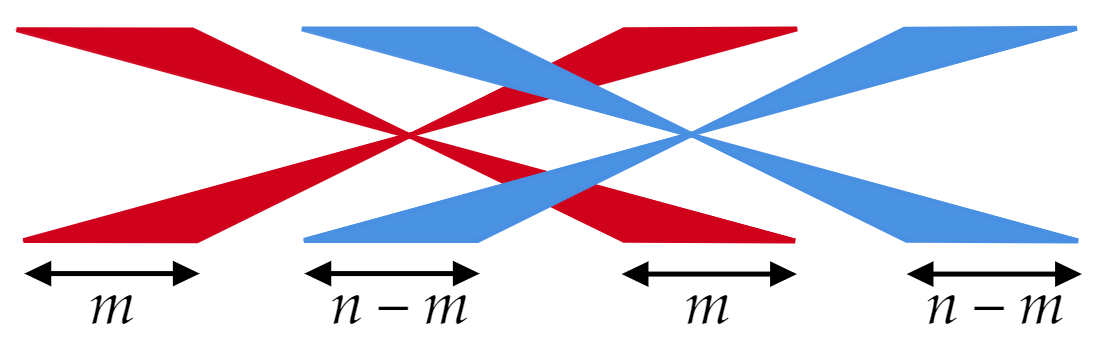}}
\end{equation}
With $\sigma_m$, the trace diagram in Fig. \ref{butterfly_cascade} has value
\begin{align}
    &Tr(U^{x_1-y_{n-m+1}})Tr(U^{x_2-y_{n-m+2}})\cdots Tr(U^{x_{m}-y_{n}}) Tr(U^{x_{m+1}-y_{1}})\nonumber \\
    &Tr(U^{x_{m+2}-y_{2}})\cdots Tr(U^{x_{n-1}-y_{n-1-m}})Tr(U^{x_{n}-y_{n-m}})
\end{align}
For the Haar average of the trace diagram to contribute at leading order ($\tilde{G}=q^{N/2}$), all the indices must cancel, this is expressed below
\begin{equation}
    q^{-N/2}\int dU G(\boldsymbol{x},\boldsymbol{y},\sigma_m) = \delta^{\boldsymbol{y},\alpha_{m+1}(\boldsymbol{x})} + \order{1/q}
\end{equation}
where the Kronecker delta checks that the $\boldsymbol{y}$ and $\alpha_{m+1}(\boldsymbol{x})$ are identical strings
where $\alpha_{n}$ is the cyclic permutation $\alpha_k=(k,k+1,\cdots,n,1,2,\cdots,k-1)$.
Therefore, the Haar average can be simplified further,
\begin{equation}
    \int dU \langle \mathcal{Z}(\boldsymbol{t}) \rangle \langle \mathcal{Z}(\boldsymbol{t}') \rangle^* = \frac{1}{q^{2}}\sum_{m=1}^{n}\delta^{\boldsymbol{y},\alpha_{m}(\boldsymbol{x})} + \order{q^{-3}}
\end{equation}
Then, with an abuse of notation, where the Kronecker delta now checks whether its arguments are equal up to global time shifts, we write
\begin{equation}
    \int dU \langle \mathcal{Z}(\boldsymbol{t}) \rangle \langle \mathcal{Z}(\boldsymbol{t}') \rangle^* = \frac{1}{q^{2}}\sum_{m=1}^{n}\delta^{\boldsymbol{t}',\alpha_{m}(\boldsymbol{t})} + \order{q^{-3}}
\end{equation}
This is precisely what we wanted to show.
\end{proof}

\section{Out-of-time-ordered correlators (OTOCs)}\label{OTOC section}
For lack of a better name, we will refer OTOCs of the form $\langle Z A Z(T) B^\dagger Z C Z(T) D^\dagger \rangle$
where $A$, $B$, $C$ and $D$ are time-ordered products of the form $Z(1)^{a_1}\cdots Z(T-1)^{a_{T-1}}$ for $a_1\in\{0,1\}$, as `physical' OTOCs because they appear naturally in the perturbative expansions in the memory matrix calculation of \cite{MMFpaper}. Physical OTOCs can also be represented as contraction of a tensor $\Gamma=\Gamma(1)\cdots\Gamma(T-1)$ which is given by the product of $T-1$ tensors, each with 4 input and output legs ($1,\overline{1},2,\overline{2}$) as shown below
\begin{figure}[H]
	\centering
	\large
	$\textrm{OTOC}_\Gamma =$ \Large $\frac{1}{q}$ \raisebox{-0.45\totalheight}{\includegraphics[height = 2cm]{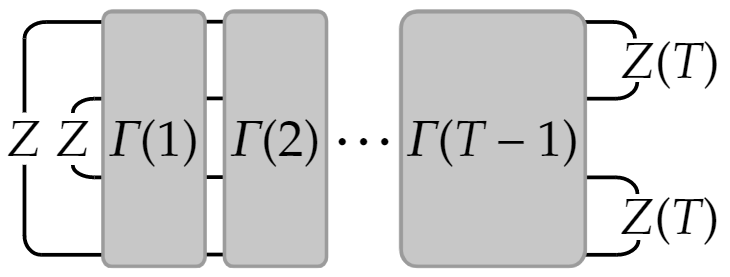}}
	\caption{A physical OTOC, each layer $\Gamma(j)$ represents possible insertions of operators $Z(j)$ on legs $1$ and $2$ and operators $Z(j)^*$ on legs $\overline{1}$ and $\overline{2}$. From top to bottom, the leg ordering is $1,\overline{1},2,\overline{2}$. \label{physical OTOC}}
\end{figure}
where vertical layer $\Gamma(j)$ has the option of placing the operator $Z(j)$ on legs $1$ and $2$ and the operator $Z(j)^*$ on legs $\overline{1}$ and $\overline{2}$. We define $\ket{+}$ and $\ket{-}$ as the following wirings,
\begin{equation}\label{plusminusdef}
    \ket{+} = \frac{1}{q}\ \raisebox{-0.45\totalheight}{\includegraphics[height = 1.6cm]{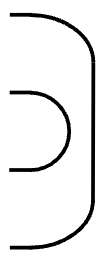}}, \quad \ket{-} = \frac{1}{q} \raisebox{-0.45\totalheight}{\includegraphics[height = 1.6cm]{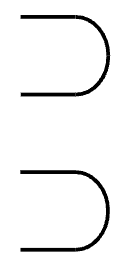}}.
\end{equation}
Let us next define $\bra{+}$ and $\bra{-}$ as the horizontal reflections of $\ket{+}$ and $\ket{-}$. A closed loop represents the trace $\Tr(\mathbb{1})$ and so has the value $q$, therefore, $\ket{\pm}$ are properly normalised, $\bra{\pm}\ket{\pm}=1$. Define $\bra{Z_+}$ ($\ket{Z(T)_-}$) as the $\bra{+}$ ($\ket{-}$) wiring with $Z$ ($Z(T)$) operators decorating both wires. With these definitions, we can rewrite the OTOC as
\begin{equation}
    \textrm{OTOC}_\Gamma = q\bra{Z_+}\Gamma\ket{Z(T)_-}
\end{equation}
Finally, define the projector
\begin{equation}\label{Kdef}
    K=\frac{1}{1-q^2}\left(\ket{+}\bra{0} + \ket{-}\bra{\perp}\right),
\end{equation}
where $\bra{0}=\bra{+}-\frac{1}{q}\bra{-}$ and $\bra{\perp}=\bra{-}-\frac{1}{q}\bra{+}$. Checking that $K$ is a projector is a simple task, and is omitted here.

\begin{theorem}\label{OTOCs}
    The Haar average of a physical OTOC, $\textup{OTOC}_\Gamma=q\bra{Z_+}\Gamma\ket{Z(T)_-}$, is found to leading order in $1/q$ by inserting the projector $K$ either side of every layer $\Gamma(j)$,
    \begin{equation}
        \int dU \textrm{OTOC}_\Gamma = q\bra{Z_+}K\Gamma(1)K\cdots K\Gamma(T-1)K\ket{Z(T)_-} + \order{1/q^3} \quad \textrm{as $q\to\infty$}.
    \end{equation}
\end{theorem}
\begin{proof}
We have already proven in section \ref{correlator} that the Haar average of a single correlation function is bounded by $C/q^2$ in the large $q$ limit for some $q$ independent constant $C$. We aim to identify which physical OTOCs are $\order{1/q^2}$ and what the precise contribution is. A generic physical OTOC is shown in the figure below
\begin{equation}
	\raisebox{-0.52\totalheight}{
	\includegraphics[height = 2cm]{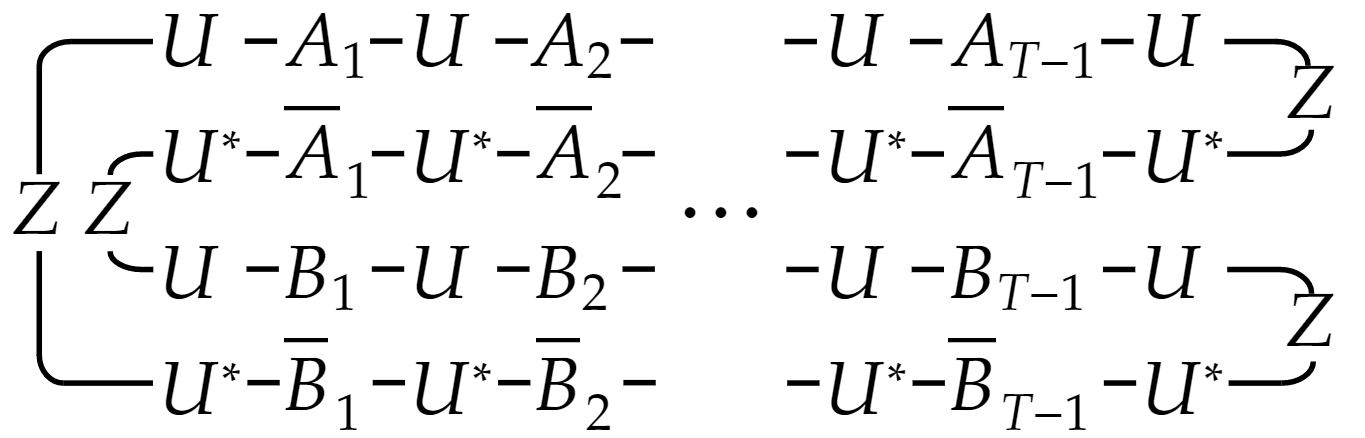}},
\end{equation}
where $A_i,\overline{A}_i,B_i,\overline{B}_i\in \{\mathbb{1},Z\}$. Writing its Haar average as a sum over permutations with Weingarten weights (without making use of an auxiliary unitary as done in the previous proofs), we find
\begin{equation}
    \int dU \textrm{OTOC} = \sum_{\sigma,\tau\in S_{2T}}\frac{Wg(\sigma\tau^{-1})}{q} \mathcal{G}(\sigma,\tau),
\end{equation}
where the number of instances of $U$ appearing in the OTOC (sum of the positive powers of $U$ appearing in the OTOC) is $2T$ and where $\mathcal{G}(\sigma,\tau)$ is given by the trace diagram below. The trace is made clear by the repetition of leg labels on the incoming and outgoing legs,
\begin{equation}
	\mathcal{G}(\sigma,\tau)=
	\raisebox{-0.52\totalheight}{
	\includegraphics[height = 4cm]{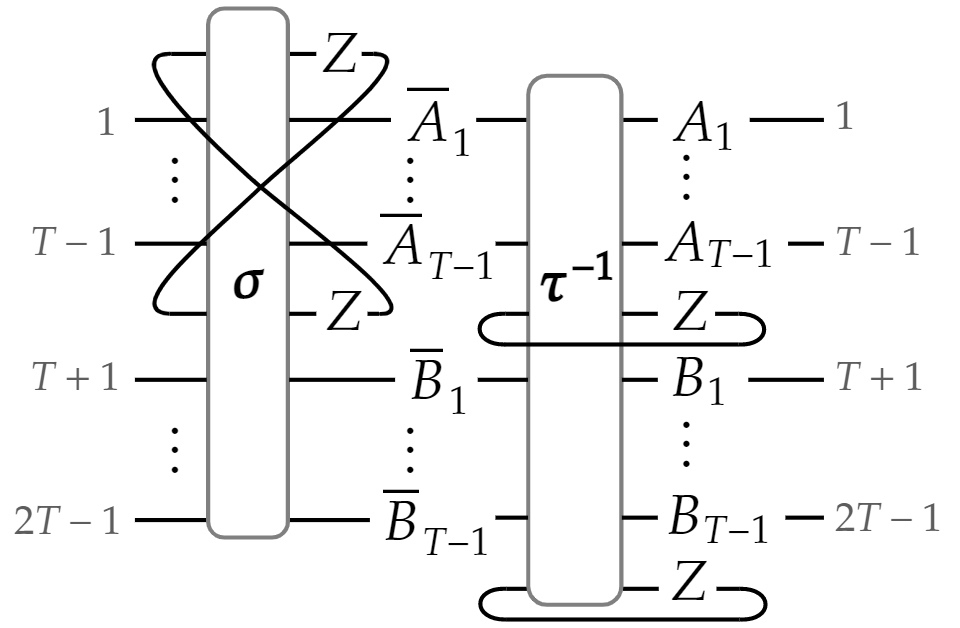}}.
\end{equation}
Bounding $\mathcal{G}(\sigma,\tau)$ takes a little work. Focus on the permutation $\sigma$ and write down the qualitatively distinct options for the first and $T$-th legs.
\begin{figure}[H]
	\centering
	\includegraphics[height = 3cm]{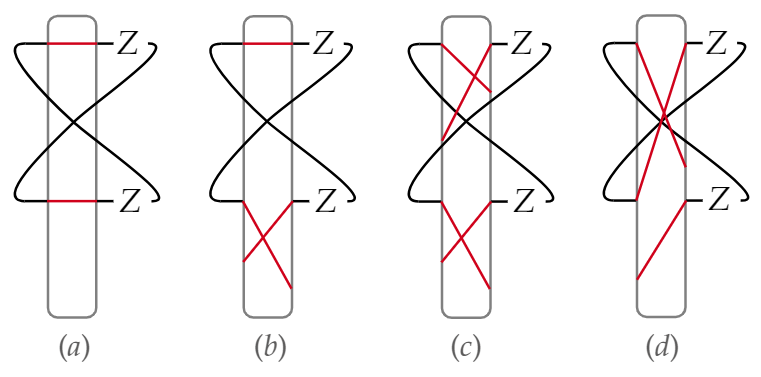}
	\caption{(a) $\sigma$ contains the $1$-cycles $(1)(T)$; (b) $\sigma$ contains a $1$-cycle on only one of the legs, $(1)$ or $(T)$; (c) $\sigma$ doesn't contain a $1$-cycle on either leg $1$ or $T$, neither does it contain a connection $1\to T$ or $T\to 1$; (d) $\sigma$ contains either the connection $1\to T$ or $T\to 1$ or both.\label{wiring_options_sigma_perm}}
\end{figure}
Assuming option (a) in Fig. \ref{wiring_options_sigma_perm}, the closed loop in the $\sigma$ block contributes a factor $q$, the remainder of $\sigma$ is given by some wiring from the $2T-2$ remaining input legs to the $2T-2$ remaining output legs. Likewise, with option (b) the $\sigma$ block is given by some wiring between the $2T-2$ remaining incoming and outgoing legs but this time without the additional closed loop. With option (c), much the same as (b), the $\sigma$ block is given by some wiring between $2T-2$ remaining input and output legs, however two of these output legs pick up a $Z$ decoration. Finally, with option (d), the diagram is zero as this involves taking the trace of a $Z$ decoration. We outline the $\tau$ block counterpart in the figure below.
\begin{figure}[H]
	\centering
	\includegraphics[height = 3cm]{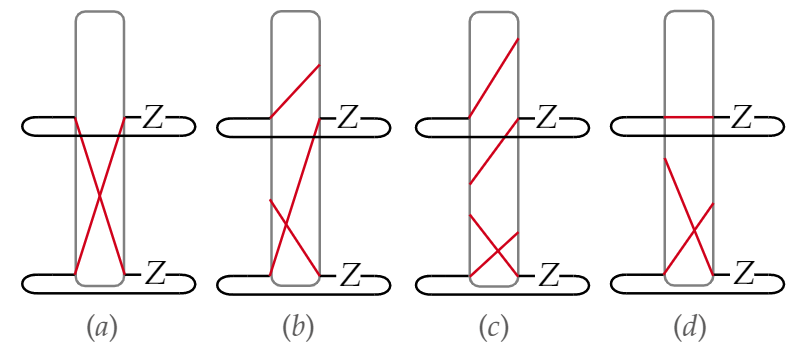}
	\caption{(a) $\sigma$ contains the $2$-cycles $(T,2T)$; (b) $\sigma$ contains the connection $T\to 2T$ or $2T\to T$ but not both; (c) $\sigma$ doesn't contain either connections $T\to 2T$ or $2T\to T$, neither does it contain either of the $1$-cycles $(T)$ or $(2T)$; (d) $\sigma$ contains either or both of the $1$-cycles $(T)$ or $(2T)$.\label{wiring_options_tau_perm}}
\end{figure}
 Assuming option (a) in Fig. \ref{wiring_options_tau_perm}, the $\tau$ block contributes an undecorated loop (a factor of $q$) and some wiring from the $2T-2$ remaining input legs to the $2T-2$ remaining output legs. With option (b), the $\tau$ block just contributes some wiring between the $2T-2$ remaining input and output legs. With option (c), the $\tau$ block also contributes some wiring between the $2T-2$ remaining input and output legs, but two of the output legs receive additional $Z$ decorations. Option (d) involves the trace of a $Z$ decoration and is therefore zero.
 
 The maximum number of loops is found if all the $2T-2$ incoming legs return to their original positions after the $\sigma$ and $\tau$ block wirings. In this case they contribute $2T-2$ loops, if we choose option (a) in both the $\sigma$ and $\tau$ blocks we find two additional loops and have an upper bound of $2T$ loops total, whereas all other options give at most $2T-1$ loops. Let $|\sigma\tau^{-1}|\geq2$, then even with the upper bound of $2T$ loops (a factor $q^{2T}$ from $\mathcal{G}$), the factor $|\frac{Wg}{q}|\sim q^{-2T-1-|\sigma\tau^{-1}|}$ gives a combined contribution of order $\order{q^{-3}}$ if $|\sigma\tau^{-1}|>2$. Therefore the $\order{q^{-2}}$ contributions to the Haar averaged OTOC must come from contributions with $\sigma=\tau$ or $|\sigma\tau^{-1}|=1$. We will study each case separately. This simplifies the OTOC Haar average as shown below,
 \begin{equation}
    \int dU \textrm{OTOC} = \langle OTOC \rangle_{+} - \langle OTOC \rangle_{-} + \order{1/q^3}.
\end{equation}
where
\begin{equation}
    \langle OTOC \rangle_{+} = \frac{1}{q^{2T+1}}\sum_{\sigma\in S_{2T}} \mathcal{G}(\sigma), \quad \langle OTOC \rangle_{-} = \frac{1}{q^{2T+2}}\sum_{\substack{\sigma,\tau\in S_{2T} \\ |\sigma\tau^{-1}|=1}} \mathcal{G}(\sigma,\tau)
\end{equation}
and where we have used the shorthand $\mathcal{G}(\sigma)=\mathcal{G}(\sigma,\sigma)$.
 
 \subsection[\texorpdfstring{tex}{pdfbookmark}]{Contributions with $\sigma=\tau$: $\langle \textrm{OTOC} \rangle_{+}$}

 \begin{equation}
     \langle \textrm{OTOC} \rangle_{+} = \frac{1}{q^{2T+1}}\sum_{\sigma\in S_{2T}} \mathcal{G}(\sigma)
 \end{equation}
 While we already have a diagrammatic representation for $G(\sigma)$, there is a more convenient diagrammatic representation. With $\sigma=\tau$, each instance of a unitary $U$ is paired with a $U^\dagger$ in the following manner:
 \begin{figure}[H]
	\centering
	\includegraphics[height = 2.6cm]{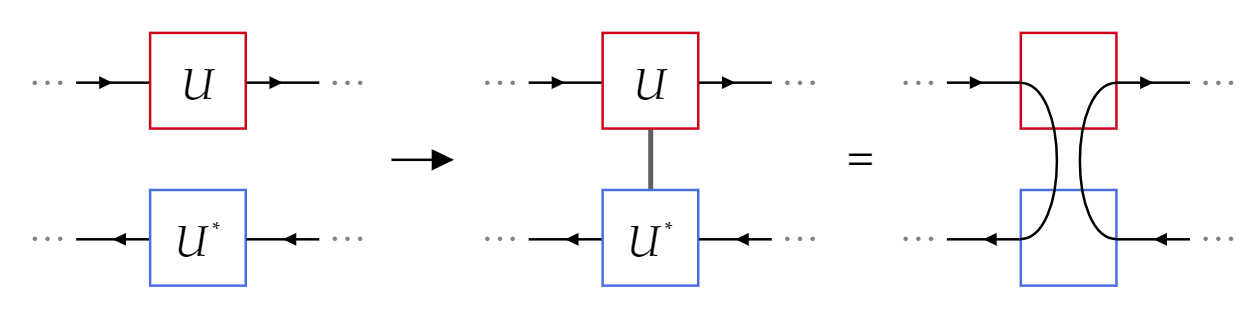}.
	\label{U_Ustar_pairing_rule}
\end{figure}
After choosing a pairing $\sigma$ of unitaries $U$ with $U^\dagger$ (or $U^*$ if the incoming/outgoing legs have been switched in the tensor diagram, see Fig. \ref{transposition convention}), we draw a grey edge between each pair. This grey edge is a shorthand for a pair of edges which identify incoming and outgoing legs. The arrow is just cosmetic and follows the clockwise direction around the OTOC contour, as seen in Fig. \ref{OTOC_as_trace_diagram}.
 \begin{figure}[H]
	\centering
	\includegraphics[height = 1.5cm]{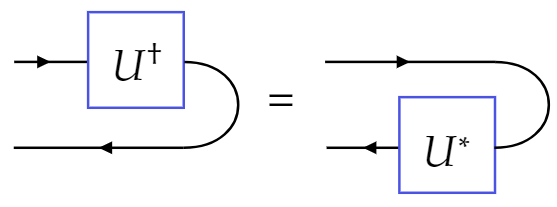}
	\caption{We use the following shorthand for moving operators around a contour, switching the input and output leg order, this introduces a transposition. In this case, changing a $U^\dagger$ into a $U^*$.}\label{transposition convention}
\end{figure}
A shorthand for the $U$-$U^*$ pairing is given below,
\begin{equation}\label{U_Ustar_pairing_rule_shorthand}
	\centering
	\includegraphics[height = 1.2cm]{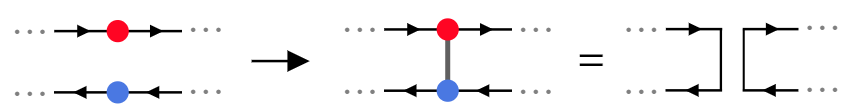}.
\end{equation}
We proceed by considering a generic physical OTOC and stretching OTO contour out into a ring,
\begin{figure}[H]
	\centering
	\includegraphics[height = 2.2cm]{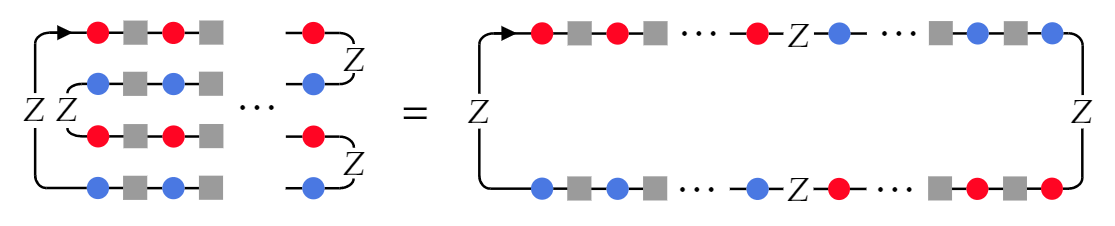}.
	\caption{A physical OTOC, the red dots represent the unitary $U$ and blue dots $U^*$ (or $U^\dagger$ depending on whether they have been brought around a bend). The grey boxes each represent either an identity operator or a $Z$ operator.}\label{OTOC_as_trace_diagram}
\end{figure}
We find that this ring is divided into quarter arcs, with the first and third being carrying the $U$'s (red dots) and the second and fourth carrying the $U^\dagger$'s (blue dots). Each arc contains the same number of coloured dots, $T$. We simplify this further by suppressing every grey box (each of which can be a $\mathbb{1}$ or $Z$). 
\begin{figure}[H]\label{OTOC_pairing_diagram}
	\centering
	\includegraphics[height = 5cm]{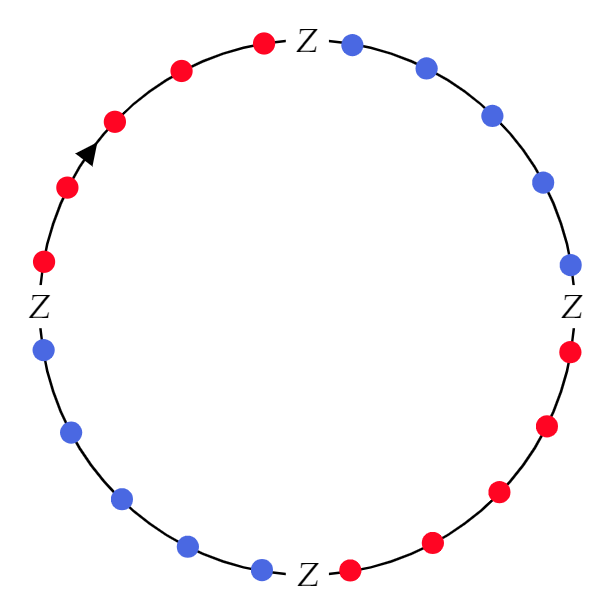}
	\caption{A physical OTOC. We have suppressed every non $U$ or $U^\dagger$ operator except for the paired $Z$ decorations at the beginning and end of every physical OTOC. The arrow signals the clockwise direction around the OTOC contour.}
\end{figure}
 A contribution $\mathcal{G}(\sigma)$ to $\langle OTOC \rangle_{+}$ is then simply a choice of pairings such as the one given below, where we have introduced an additional shorthand, that simply colours the arcs red or blue (to represent which arc contains $U$'s or $U^\dagger$'s). The positions of the unitaries are given by the vertices of the internal grey edges with the boundary ring.
 \begin{figure}[H]\label{OTOC_with_crossings}
	\centering
	\includegraphics[height = 5cm]{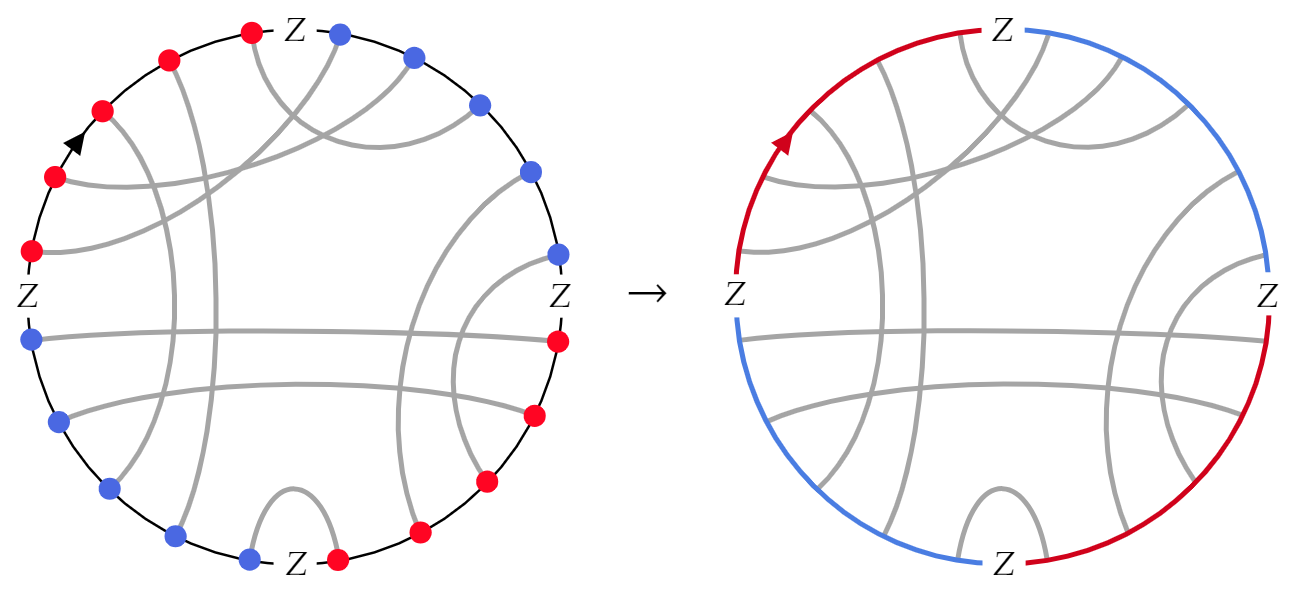}
	\caption{Example of a choice of $U$-$U^\dagger$ pairings $\sigma$. We also introduce another shorthand which drops the $U$ and $U^\dagger$ nodes and makes them implicit at the vertices of an interior edge and the perimeter edge, which is now coloured depending on whether the vertices are $U$'s or $U^\dagger$'s.}
\end{figure}
In the remainder of this section we will argue that minimising the number of internal edge crossings  optimises the number of loops, thus maximising the contribution of the diagram. To do this, let us abstract these diagrams further and consider a diagram that is only a ring with interior edges between points on the ring and with no restrictions on the positions of the vertices, such as the restriction above that edges must connect arcs with different colours. We also suppress the $Z$ decorations as these are not important in the simple task of counting the number of loops there are in each diagram. We will refer to these abstracted diagrams as \textit{cobweb diagrams}. An example of one of these cobweb diagram is given below
 \begin{figure}[H]\label{abstract_graph}
	\centering
	\includegraphics[height = 3cm]{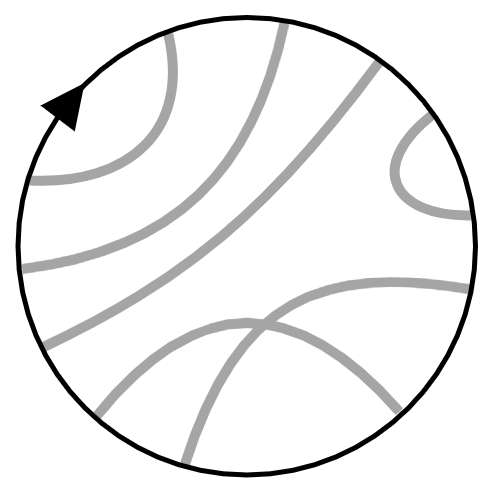}
	\caption{A cobweb diagram with interior edges that connect two points on the perimeter ring. This diagram has $E=6$ internal edges, $V=12$ vertices and $C=1$ crossing.}
\end{figure}
The evaluation of these diagrams is done with an abstracted version of Eq. \ref{U_Ustar_pairing_rule_shorthand}, which tells us how to interpret interior edges,
\begin{equation}\label{abstracted_rule}
    \raisebox{-0.45\totalheight}{\includegraphics[height = 1.2cm]{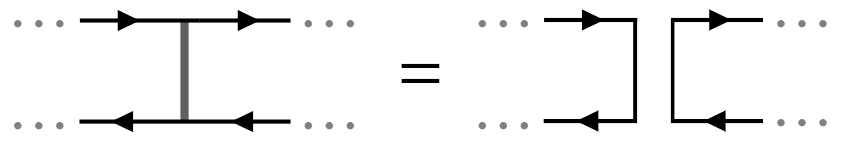}}.
\end{equation}
Using this rule, each cobweb diagram reduces to a number of loops, that we will refer to as index loops, each of which is associated with a factor $q$. We aim to find bounds on the number of index loops a cobweb diagram can have, depending on properties of the diagram such as the presence of crossings of interior edges. Using Eq. \ref{abstracted_rule}, we can identify two simple ways of extracting index loops from a cobweb diagram,
\begin{align}\label{abstracted_rules}
	\centering
	&\text{\textbf{Rule 1 (parallel-edge rule)}: } \raisebox{-0.45\totalheight}{\includegraphics[height = 1.8cm]{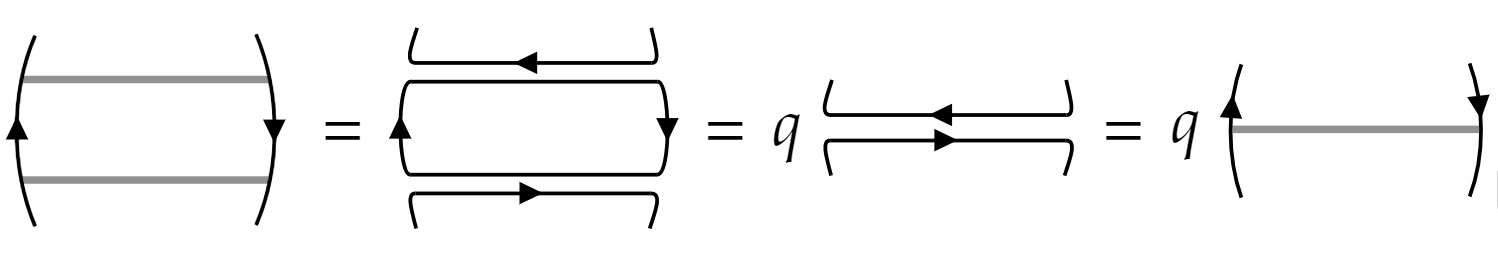}}\\
	&\text{\textbf{Rule 2 (edge-bubble rule)}: }  \raisebox{-0.45\totalheight}{\includegraphics[height = 1.6cm]{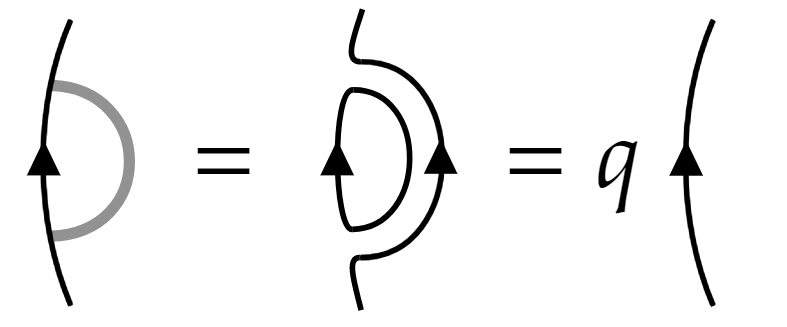}}.
\end{align}
We refer to these rules as the \textit{parallel-edge} rule and the \textit{edge-bubble} rule respectively. Each time one of these rules is applied, an index loop (or factor $q$) is extracted from the diagram. Cobweb diagrams without any edge crossings (planar cobweb diagrams) are fully reducible to the perimeter ring with no internal edges (the \textit{empty} cobweb diagram) through these rules alone. The empty cobweb is simply a single index loop and contributes one additional factor of $q$. With $E$ being the number of internal edges, planar cobweb diagram have value $q^{E+1}$ (i.e $E+1$ index loops).

However, one cannot completely evaluate all cobweb diagrams using these rules alone, the rules do not tell us how to deal with edge crossings. Rather than give a complete algorithm for determining the number of loops from a cobweb diagram, we will produce some useful bounds. 

Given some cobweb diagram $G$ with $E$ internal edges, we exhaustively apply the parallel-edge and edge-bubble rules until we arrive at what we call a \textit{reduced} cobweb diagram, where neither rule can be applied any further. Let $E_p$ be the total number of edges removed through applications of these rules. If the $E_p=E$, then the diagram has had every edge removed and we arrive at the empty cobweb diagram. Otherwise the reduced diagram has $E'=E-E_p$ internal edges;
all of which cross at least one other edge.

It will be helpful to enumerate some of ways of forming an index loop. We will focus on short index loops, i.e index loops containing only a small number of boundary edge segments. A boundary edge segment is the segment of the boundary ring bordered at each end by adjacent vertices, the total number of boundary segments is twice the number of internal edges. To form an index loop with only one boundary edge segment, start by picking a boundary edge segment. Then, to form an index loop which does not contain further boundary segments, one must draw an internal edge between the vertices connected by the boundary segment. These index loops are precisely those removed by the edge-bubble rule.

To form an index loop which contains two boundary segments, start by choosing two segments. If the boundary edges share a vertex (i.e are adjacent) it is not possible to form an index loop that contains only these two segments. If they are separated, then there are two ways to connect up these vertices (without forming the one-boundary-segment loops already considered previously), these are shown below:
\begin{enumerate}
    \item non-crossing edges: \raisebox{-0.41\totalheight}{\includegraphics[height = 1.6cm]{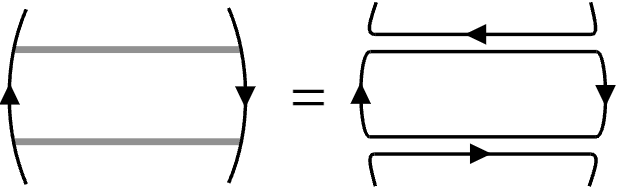}}
    
    Here we find an index loop that contains only two boundary segments.
    \item crossed edges: \raisebox{-0.41\totalheight}{\includegraphics[height = 1.6cm]{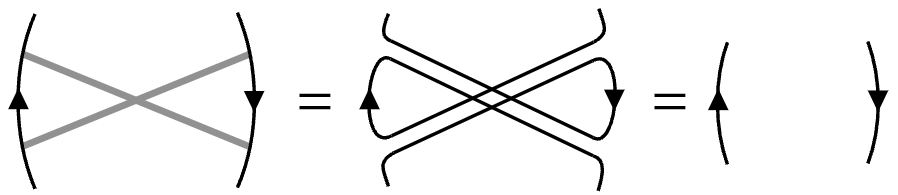}}
    
    We do not extract an index loop when the internal edges are crossed. In fact, the edges can be deleted without any change to the value of the diagram. 
\end{enumerate}
Therefore, index loops that involve exactly two boundary segments (trace a path through exactly two boundary segments) are the result of a pair of non-crossing edges of the form (1) above and are handled by the parallel-edge rule. While index loops passing through only one boundary segment are the result of edge-bubbles only and are dealt with using the edge-bubble rule. All other index loops must trace a path through three or more boundary segments. 

We can use this information to bound the number of index loops in a reduced diagram. We have just seen that the parallel-edge rule and the edge-bubble rule remove all the one-boundary-segment and two-boundary-segment index loops from the diagram. The subsequent reduced diagram therefore contains index loops each of which contain three or more boundary segments. The number of internal edges in a reduced diagram is $E'$, the number of vertices is $V'=2E'$, the number of boundary segments is $B'=V'=2E'$. As each boundary segment is part of only one index loop, then for $E'\geq 2$ ($E'=1$ is no possible), the number of index loops is bounded by $N'\leq B'/3=2E'/3$. For $E'=0$, the number of index loops is given by $N'=1$. We can combine these cases into the single inequality
\begin{equation}
   N'\leq \max(1,2E'/3) 
\end{equation}
The number, $N$, of index loops in the fully diagram is then bounded by 
\begin{itemize}\label{bounds}
    \item planar cobweb diagrams ($E'=0$): $N = E+1$
    \item non-planar cobweb diagrams ($E'\geq 2$): $N\leq E_p+\frac{2}{3}(E-E_p)=E-E'/3$
\end{itemize}
where $E$ is the total number of internal edges and $E_p=E-E'$ is the number of times the parallel-edge and edge-bubble rules were applied to arrive at the reduced diagram. A reduced diagram must have at least one crossing so $E'\geq 2$. The $E'=2$ reduced diagram is unique and contains a single index loop and therefore $N=(E-2)+1=E-1$. With $E'=3$, $N$ is bounded by $N\leq E - 1$. For $E\geq 4$, the number of loops is bounded by $N\leq E-2$ (where we have used the fact that $N$ must be an integer). Therefore, combining the factor $\textrm{Wg}(\sigma)/q$, $E'\geq 4$ cobweb diagrams contribute at order $\order{1/q^3}$ or smaller. We will now study planar cobweb diagrams, $E'=2$ and $E'=3$ cobweb diagrams in detail.

\subsubsection{Planar diagrams}

In this section, we redecorate planar cobweb diagrams with the red and blue contour colouring and also the four compulsory $Z$ decorations at the interfaces of each coloured contour. By doing this, we investigate the contribution of the planar cobweb diagrams to $\langle \textrm{OTOC} \rangle_{+}$. The important distinction of the diagrams with $U$ (red) and $U^*$ (blue) contours is that there is an added condition that edges must pair points on contours of differing colour.

It is not possible to draw a planar diagram without a bubble edge, this bubble edge must either: (1) straddle a $Z$ decoration at the interface of a red and blue contour, and hence take the trace $\Tr(Z)=0$ as shown in Fig. \ref{NN_edge_straddling_Z}; or (2) connect adjacent vertices within a single contour, however, this is not permitted as edges must connect vertices in contours of different colour.

\begin{figure}[H]\label{NN_edge_straddling_Z}
	\centering
	\includegraphics[height = 2cm]{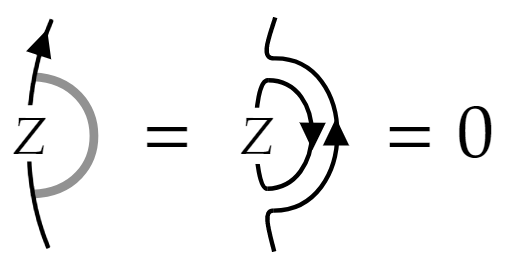}
\end{figure}

Therefore, no planar cobweb diagrams contributes to the physical OTOCs.

\subsubsection[\texorpdfstring{tex}{pdfbookmark}]{$E'=3$ cobweb diagrams}
The only reduced diagram with $E'=3$ edges is shown below,
 \begin{figure}[H]\label{6_vertex_loop_special_case}
	\centering
	\includegraphics[height = 2.5cm]{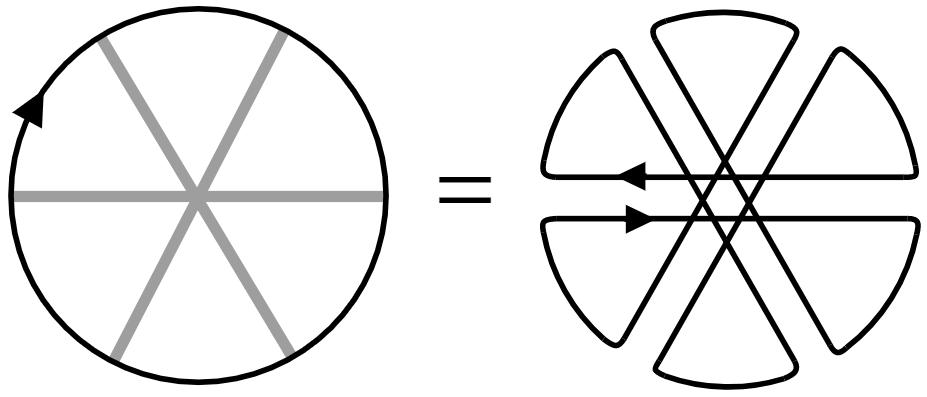}
\end{figure}
We can apply the parallel-edge and edge-bubble rules in reverse in order add back in edges and construct every diagram that reduced to this $E'=3$ diagram. We do this to build a diagram $G$, but keep the three initial edges highlighted in our minds. We then reintroduce the $Z$ decorations and the coloured $U$ and $U^\dagger$ contours and find that there are two qualitatively different types of $E'=3$ diagram. In the first, only two contours are connected by the highlighted edges, in the second, three contours are connected by the highlighted edges. Below we show these two cases but only display the highlighted edges in order to de-clutter the diagrams.
 \begin{figure}[H]
	\centering
	\includegraphics[height = 3.4cm]{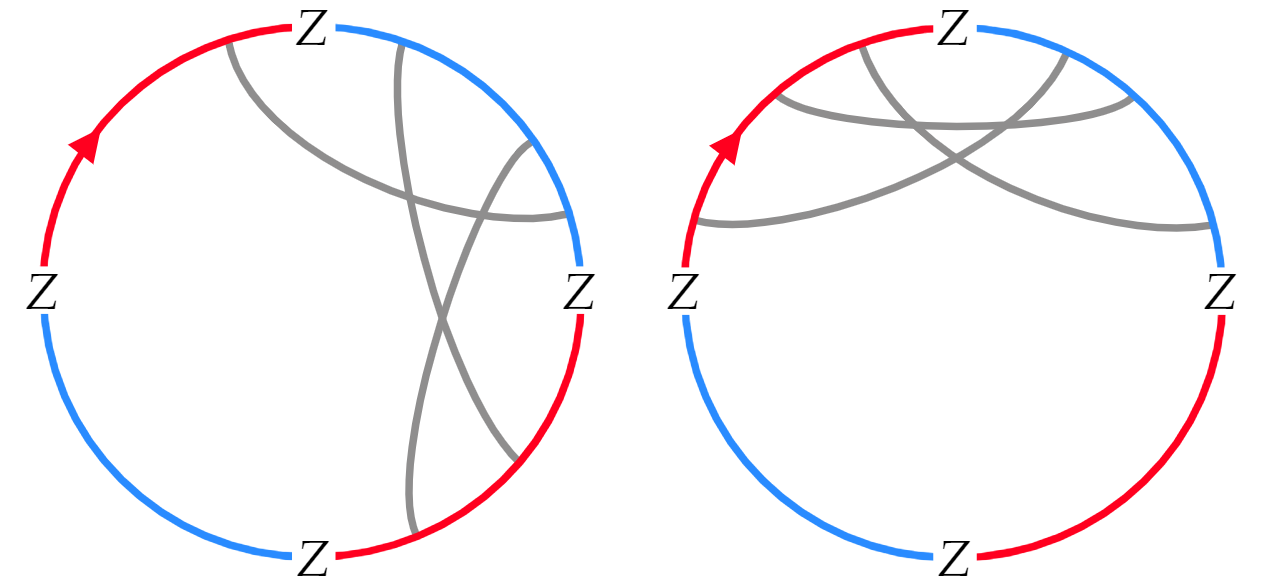}
	\caption{The two qualitatively different $E'=3$ diagram. (Left) Three contours are connected by the highlighted edges. (Right) Only two contours are connected by the highlighted edges.\label{6_vertex_OTOC}}
\end{figure}
The edges that have not been displayed must either be parallel to another edge (in the sense used in the parallel edge rule) or be an edge-bubble. The edge-bubble option, as in the case of the planar diagrams, will lead to the trace of a $Z$ decoration and is therefore not allowed. In the first case (on the left of Fig. \ref{6_vertex_OTOC}), it is not possible to add edges connecting to lower left blue contour to one of the red contours using only the reversed parallel-edge. Therefore, the lower left blue contour must have no vertices in the full diagram $G$. However, another condition of the coloured diagrams is that each must have the same number of vertices, this excludes the diagram $G$ just discussed. Next, consider the second case (on the right of Fig. \ref{6_vertex_OTOC}), it is not possible to add edges that connect to either of the lower two contours using only the reversed parallel-edge rule. Using again, the fact that each contour must have the same number of vertices, we are forced to dismiss this second type of $E'=3$ diagram. Therefore, there are no $E'=3$ diagrams that contribute to the Haar averaged OTOC. This leaves only the $E'=2$ contributions to consider.

\subsubsection[\texorpdfstring{tex}{pdfbookmark}]{$E'=2$ circle graphs}
The only $E'=2$ reduced diagram is given below.
    \begin{equation}\label{special case}
       \raisebox{-0.46\totalheight}{\includegraphics[height = 2cm]{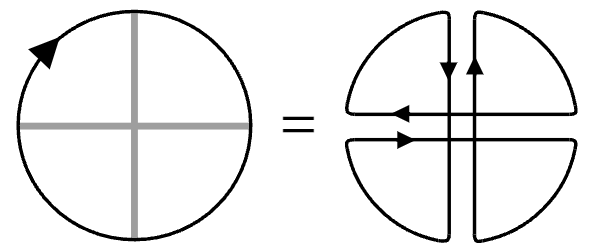}} 
    \end{equation}
As before, by using the parallel-edge and edge-bubble rules in reverse, we can build every cobweb diagram that reduces to a $E'=2$ reduced diagram. Doing this we build a diagram $G$, keeping the initial two edges highlighted in our minds before reintroducing the coloured contours and the $Z$ decorations at each of the four interfaces. This gives two qualitatively distinct types of diagram which are shown below (where only the highlighted edges are drawn).

 \begin{figure}[H]
	\centering
	\includegraphics[height = 3.4cm]{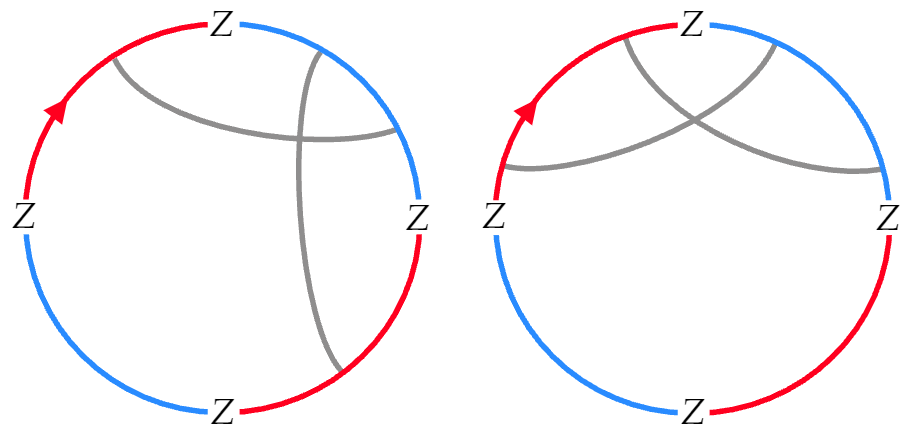}
	\caption{The two qualitatively different $E'=2$ diagram. (Left) Three contours are connected by the highlighted edges. (Right) Only two contours are connected by the highlighted edges.}\label{Eis2diagrams}
\end{figure}
As in the case of $E'=3$ diagrams, any edges added using the edge-bubble rule in reverse would have to straddle a $Z$ decoration and therefore take the trace $\Tr(Z)=0$. Consider the diagram on the right of Fig. \ref{Eis2diagrams}, no edge can be added that connects to either of the lower two contours using the reversed parallel-edge rule. Therefore, by requiring each contour has the same number of vertices, this type of $E'=2$ diagram is excluded. This leaves only the diagram on the left of Fig. \ref{Eis2diagrams}. Adding edges using the parallel-edge rule in reverse generates a family of diagrams shown below, parameterised by four numbers $E_{1,\overline{2}}$, $E_{2,\overline{1}}$, $E_{1,\overline{1}}$ and $E_{1,\overline{1}}$, where $E_{i,\overline{j}}$ counts the number of edges connecting contour $i$ with $\overline{j}$.
 \begin{figure}[H]\label{waffle_OTOC}
	\centering
	\includegraphics[height = 4.5cm]{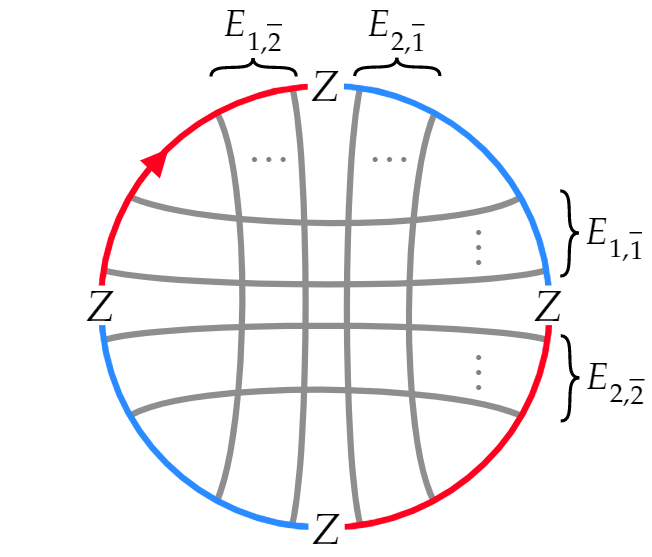}
\end{figure}
Using once more, the fact that the diagram must have equal number of vertices on all of four of the arcs, $V_1=V_{\overline{1}}=V_2=V_{\overline{2}}$, we find the condition $E_{1,\overline{1}}+E_{1,\overline{2}}=E_{1,\overline{1}}+E_{2,\overline{1}}=E_{2,\overline{1}}+E_{2,\overline{2}}=E_{1,\overline{2}}+E_{2,\overline{2}}$ in the diagram above. This solved by $E_{1,\overline{1}}=E_{2,\overline{2}}$ and $E_{1,\overline{2}}=E_{2,\overline{1}}$. Therefore, the $E'=2$ diagrams contributing to the averaged OTOC are in fact parameterised only by the total number of edges $E=2T$ and by $N_{-}$, the number of edges connecting arc $1$ to $\overline{1}$, where $1\leq N_{-}<T$. Each of these diagrams with $E=2T$ total edges has $E-1$ loops, combining the accompanying factor $\frac{1}{q^{2T+1}}$ this gives $\order{1/q^2}$ contributions to $\langle \textrm{OTOC}\rangle_+$.

Folding one of these diagrams back up into the OTOC contour reveals a useful correspondence.
 \begin{equation}\label{OTOC_with_projectors}
	\langle \textrm{OTOC} \rangle_{+} = \sum_{N_-=0}^{T} \frac{1}{q^2}
	\raisebox{-0.64\totalheight}{\includegraphics[height = 3cm]{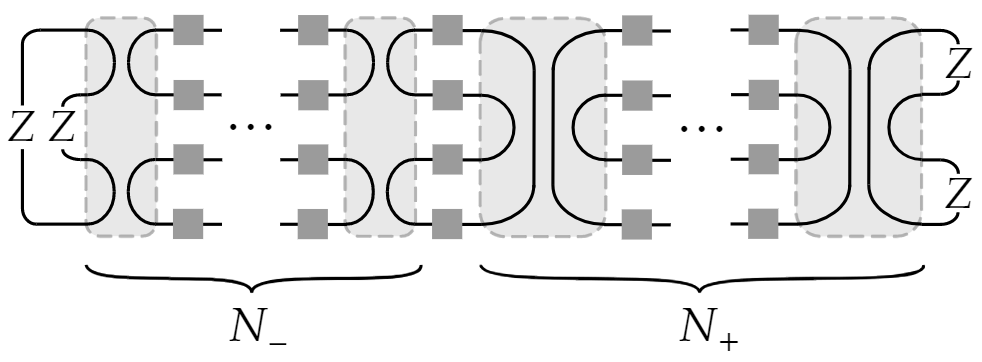}}
\end{equation}
where $N_{+} + N_{-}=T$ in the case of physical OTOCs with end time $T$. We have been a little sloppy with the factors of $q$ in this picture, this is easily resolved by associating with each closed loop a factor of $1/q$ to keep each loop normalised. Therefore, the $\order{q^{-2}}$ contributions to $\langle \textrm{OTOC} \rangle_+$ are found simply inserting the projector $K_-=\ket{-}\bra{-}$ for $N_{-}$ layers and then the projector $K_+=\ket{+}\bra{+}$ for the remaining layers and summing over $N_-$ from $N_-=1$ to $T-1$.
\subsection[\texorpdfstring{tex}{pdfbookmark}]{Contributions with $|\sigma\tau^{-1}|=1$: $\langle OTOC \rangle_{-}$}
\begin{equation}
     \langle \textrm{OTOC} \rangle_{-} = \frac{1}{q^{N+2}}\sum_{\substack{\sigma,\tau\in S_N \\ |\sigma\tau^{-1}|=1}} \mathcal{G}(\sigma,\tau)
\end{equation}
As in the $\sigma=\tau$ case, we can stretch the OTOC contour into a circle and then draw edges between $U$'s and $U^\dagger$'s. However, this time, while almost every $U$ is paired with a $U^\dagger$ as before, now two of the $U$'s and two of the $U^\dagger$'s a grouped in a cycle $U_i\to U^\dagger_j \to U_k \to U^\dagger_l\to U_i$. This means that there are two kinds of internal edges: (1) the grey edges used previously which, because they carry two indices, we now also refer to them as \textit{doubled edges}; (2) the new edge which we call a \textit{single-index edges} as they carry only one index each, these are represented as thin black edges in the interior of the boundary ring. The rules for the single-index edges are shown below, where the curved edges are boundary segments (part of the boundary ring).
\begin{figure}[H]\label{directed_edge_rule}
	\centering
	\includegraphics[height = 2cm]{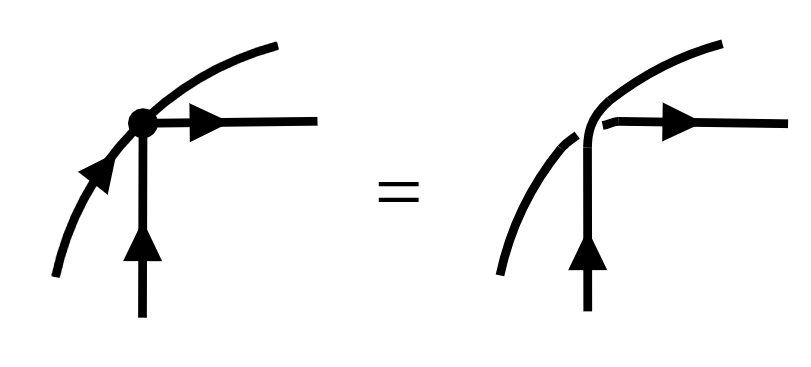}
\end{figure}
Before we consider the consequences for the diagrams with the four coloured contours (two red and two blue) and $Z$ decorations, we will consider abstracted diagrams as we did before. These abstracted diagrams are the same as before, with the addition of a cycle of four single-index edges. There are four qualitatively different realisations of a cycle of four single-index edges. In each case, the diagram fragments into sub-diagrams, possibly connected by grey (doubled) edges. Individually, these sub-diagrams are identical to the cobweb diagrams seen in the previous section. However, unlike in the previous section and because there are now multiple cobwebs at once, they can be connected by grey edges. Below we shown an example for each of the four qualitatively different ways of fragmentation. We will choose a convention in which all edges belonging to a single sub-diagram are drawn entirely within the sub-diagram's boundary ring, and all edges connected different sub-diagrams are drawn such that they never cross the interior of a sub-diagram. The usefulness of this convention is shown in section \ref{convention}

\vspace{1mm}

\begin{enumerate}
    \item
	\raisebox{-0.45\totalheight}{\includegraphics[height = 2.5cm]{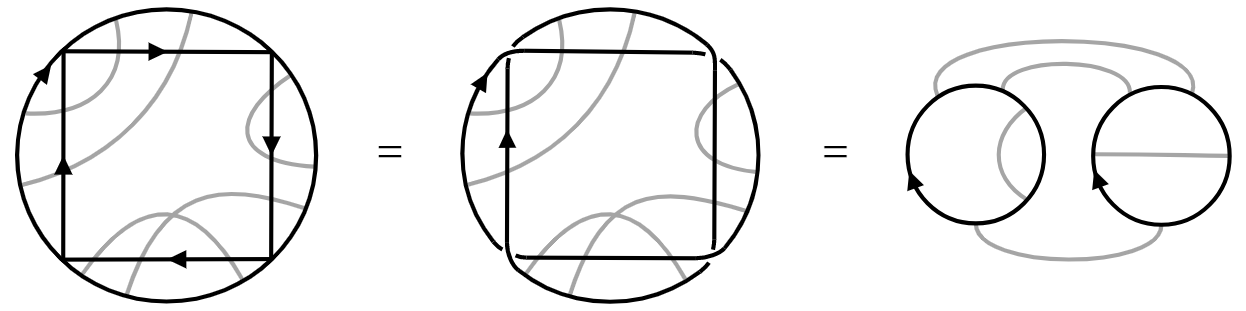}}
	\item
	\raisebox{-0.45\totalheight}{\includegraphics[height = 2.5cm]{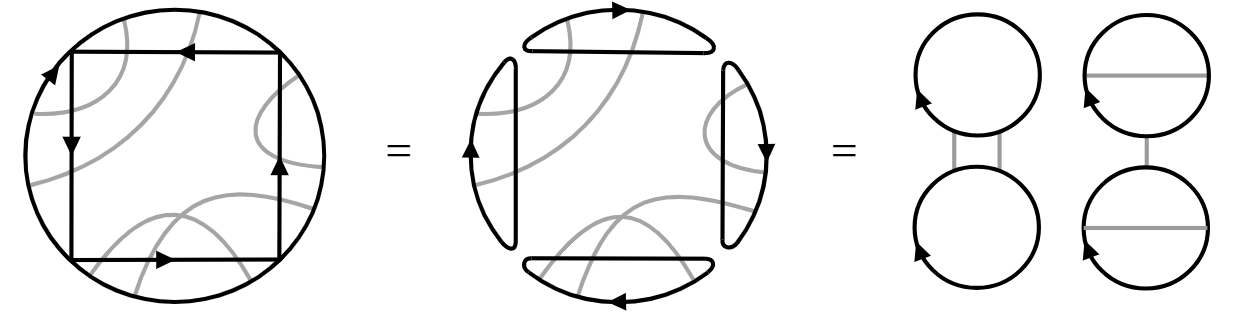}}
	\item \
	\raisebox{-0.45\totalheight}{\includegraphics[height = 2.5cm]{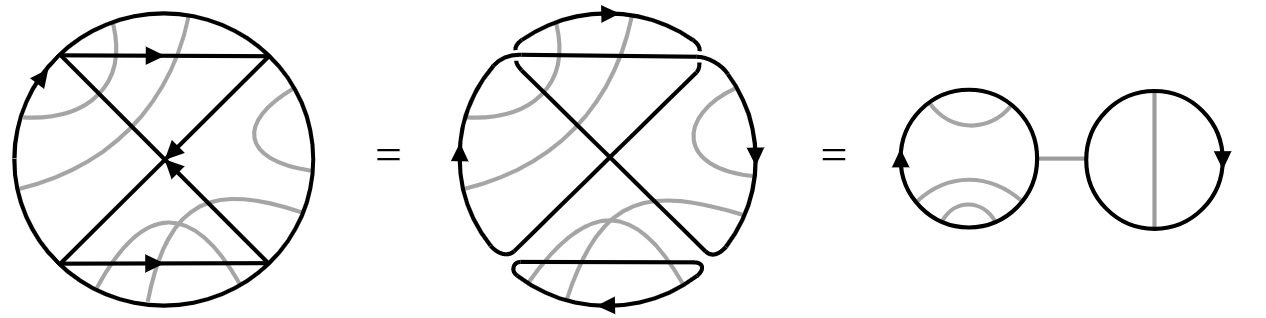}}
	\item \
	\raisebox{-0.45\totalheight}{\includegraphics[height = 2.5cm]{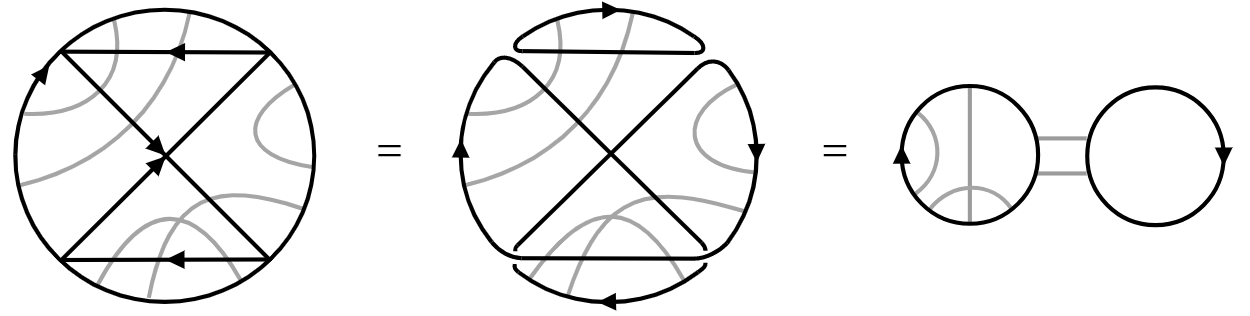}}
\end{enumerate}
As before, we will attempt to identify index loops that contain as few boundary segments as possible. The analysis is almost identical to that of the previous section. All index loops that contain only one boundary segment must be those associated with the bubble-edge rule. All index loops that contain only two boundary segments are those associated with the parallel-edge rule, except for a special case that we will discuss in a moment. The proof of the two-boundary-segment index loop case is as follows: If the boundary segments belong to the same sub-diagram then the argument in the previous section applies. If they belong to different sub-diagrams then the vertices can be connected in only three ways (excluding pairings into one-boundary-segment index loops):
\begin{enumerate}
    \item non-crossing edges: \raisebox{-0.41\totalheight}{\includegraphics[height = 1.4cm]{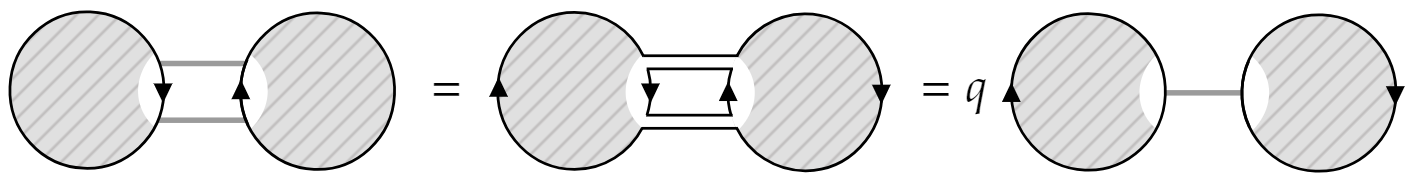}}
    
    Here we have found and eliminated (at the expense of a factor of $q$) an index looping involving two boundary segments on different sub-diagrams. This is in fact just an application of the parallel-edge rule.
    \item crossed edges: \raisebox{-0.41\totalheight}{\includegraphics[height = 1.4cm]{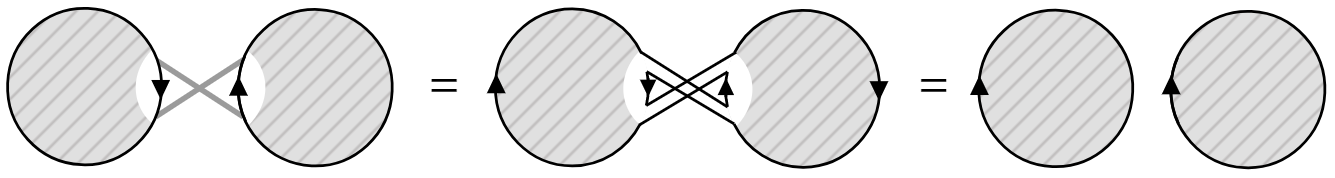}}
    
    Here we have not found an index loop. Instead we have managed to delete the crossed edges without accumulating any factors of $q$.
    
    \item  Special case: in this case we consider only a single edge between two sub-diagrams. If both of the sub-diagrams are otherwise empty, the diagram is equivalent to a single index loop, this index loops involves only two boundary segments. This is shown below.
    \begin{center}
    \raisebox{-0.41\totalheight}{\includegraphics[height = 1.4cm]{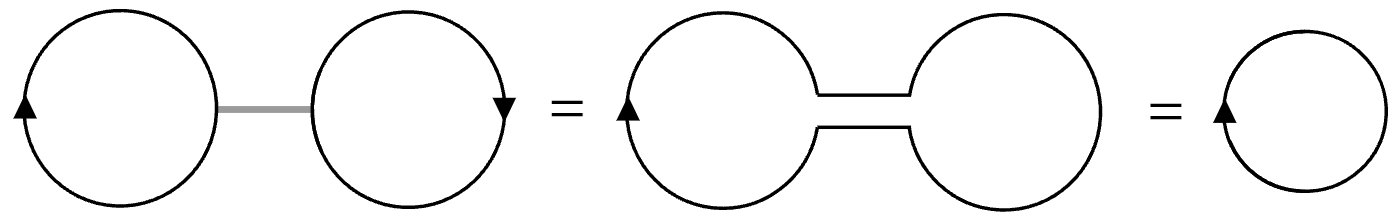}}
    \end{center}

\end{enumerate}
Therefore, two-boundary-segment index loops are either extracted using the parallel-edge rule or is found when two otherwise empty sub-diagrams are connected by a single grey edge.

\begin{tcolorbox}[colback=white]
\subsubsection{Why the convention?}\label{convention}
We have adopted a convention in which edges between sub-diagrams cannot cross the interior of any sub-diagram. If we didn't follow this convention when drawing diagrams, we would run into the following situations: (1) a pair of edges may appear to be crossing and yet the parallel-edge rule still applies, this is to do with the fact that despite the edges appearing crossed, the orientation of the arrows on the boundary edges differs from that shown in the parallel-edge rules; (2) a pair of edges that appear to be parallel are in fact dealt with by the edge crossing rule. Below we show examples of this.
\begin{figure}[H]\label{fragmented_apparent_crossing}
	\centering
	\includegraphics[height = 1.6cm]{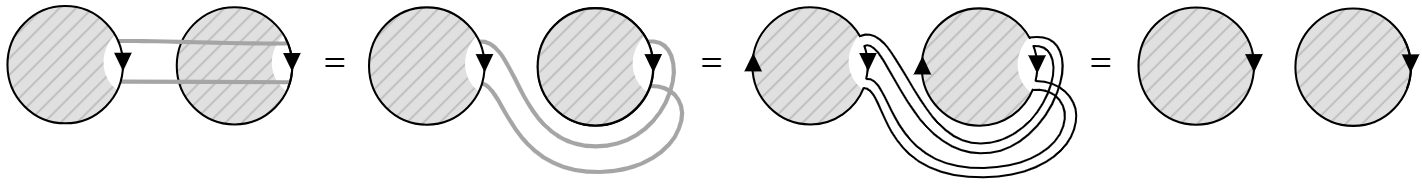}
	\\
	\centering
	\includegraphics[height = 1.6cm]{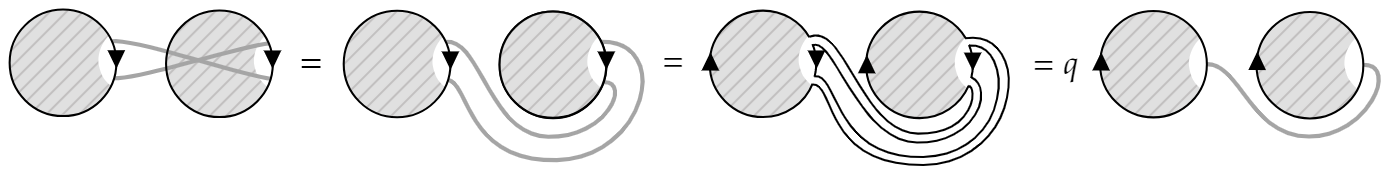}
\end{figure}
Any edges that appear parallel when drawn according to our convention can be treated with the parallel edge rule with no complications or subtleties.
\end{tcolorbox}

We can now reduce a diagram $G$ with $N_S$ sub-diagrams by using the parallel-edge rule and edge-bubble rule. The diagram $G'$ we arrive at after exhaustively applying the rules is again referred to as reduced diagram. We will use this strategy to find bounds on the number of index loops for an initial diagram $G$ consisting of $N_S=2$ or $4$ sub-diagrams, $E_i$ internal edges that belong to the $i$-th sub-diagram, and $E_{i,j}$ edges between sub-diagrams $i$ and $j$. The reduced diagram will have $E_i'$ edges belonging to $i$-th sub-diagram and $E'_{i,j}$ edges between sub-diagrams $i$ and $j$. Each application of the parallel-edge rule and the edge-bubble rule reduces the number of edges in the diagram by one. Therefore, $E'=E-E_p$, where $E=\sum_i E_i + \sum_{i,j}E_{i,j}$, $E'=\sum_i E'_i + \sum_{i,j}E'_{i,j}$ and $E_p$ is the number of times an edge was removed by the rules. An important thing to note is that the rules cannot remove all edges between two sub-diagrams, unless they initially shared no edges. This is summarised below. 
\begin{equation}
E'_{i,j} \geq 1 \textrm{ iff } E_{i,j}\geq 1, \quad
E'_{i,j} = 0 \textrm{ iff } E_{i,j}=0
\end{equation}
Let $M$ be the number of times the special case configuration appears in $G'$ and $N'_{\textrm{planar}}$ the number of disconnected planar sub-diagrams in $G'$. We are able to bound the number of index loops in $G'$ as follows: each disconnected planar sub-diagram in $G'$ is fully reduced (i.e the empty cobweb diagram) and therefore contributes a single index loop; each special case configuration contributes a single index loop; every remaining index loop in $G'$ contains at least three boundary segments, allowing us to bound the number of these index loops from above by $2(E'-M)/3$, where $E-M$ is the number of edges left after removing the special case configuration. This gives the bound

\begin{equation}\label{reduced bound}
    N'\leq N'_{\textrm{planar}}+ M + 2(E'-M)/3 = N'_{\textrm{planar}} + 2E'/3 +M/3,
\end{equation}
where $N'_{\textrm{planar}}+M/2\leq N_S$.
The total number of index loops for a diagram $G$ (that reduces to $G'$) is then given by
\begin{equation}\label{general bound}
    N\leq E-(E'-M)/3 + N_{\textrm{planar}},
\end{equation}
where we have uses the fact that the number of disconnected planar sub-diagrams is conserved through the reduction process to say
$N_{\textrm{planar}}=N'_{\textrm{planar}}$. We remind the reader that the Weingarten factor is given by $\textrm{Wg}(\sigma\tau^{-1}) \approx -q^{-2T-1}=-q^{-(E+2)-1}$, where because two of the $2T$ instances of $U$ have been already allocated single-index edges, there are only $E=2T-2$ grey (doubled) edges. The contribution made from a diagram with $N$ index loops has the size $\abs{\frac{\textrm{Wg}(\sigma\tau^{-1})}{q}q^N}\approx q^{N-E-4}$. In order for the contribution to be $\order{1/q^2}$ or larger, the number of loops must satisfy $N-E\geq 2$. Comparing this to the inequality Eq. \ref{general bound}, we see that
\begin{equation}
    \frac{M-E'}{3}+N_{\textrm{planar}}<2 \implies \textrm{ contributes at $\order{1/q^3}$ or smaller.}
\end{equation}
Importantly $M-E'\leq 0$ and $N_{\textrm{planar}}\leq N_S$. Therefore, for $N_S=2$, a contribution can only be relevant if $N_{\textrm{planar}}=2$, this accounts for both sub-diagrams and so $M=0$. Therefore, the only $N_S=2$ contributions at $\order{1/q^2}$ are those where the sub-diagrams are disconnected and planar. For $N_S=4$, the situation is not so tightly constrained. Due to the fact that $M-E'\leq0$ we can see that we need $2\leq N_{\textrm{planar}}\leq 4$. For $N_{\textrm{planar}}=2$ we require $M=E'$ and we only have two sub-diagrams to use, therefore $M\leq 1$. If $M=E'=0$, these sub-diagrams must be planar which contradicts the assumption that $N_{\textrm{planar}}=2$, therefore we must have $M=E'=1$. For $N_{\textrm{planar}}=3$ we are left with a single disconnected sub-diagram that, by assumption, is not planar, i.e $E'\geq 2$ and $M=0$. Finally, for $N_{\textrm{planar}}=4$, we obviously have $M=E'=0$.

For $N_S=4$, there are 3 configurations to consider: $(N_{\textrm{planar}},E',M)=(2,1,1)$, $(3,2,0)$ and $(4,0,0)$. While for $N_S=2$, there is only one $(2,0,0)$.

\subsubsection[\texorpdfstring{tex}{pdfbookmark}]{Reintroducing the coloured contours and $Z$ decorations}
It is now time to reintroduce the coloured contours and the $Z$ decorations at the interfaces of these contours. Unlike in the abstracted colourless case, there six qualitatively different fragmentation schemes, depending of whether the cycle of four single-index edges explore all four contours, three contours, or only two contours. The six schemes are: each $U$ and $U^\dagger$ belong to different contours and the cycle is (1) clockwise or (2) anti-clockwise; both $U$'s ($U^\dagger$'s) belong to the same contour while the $U^\dagger$'s ($U$'s) belong to different contours with different directions of the cycle being (3) and (4) as shown below; both $U$'s belong to the same contour and both $U^\dagger$'s belong to the same contour with the two directions of the cycle being (5) and (6) as shown below.
\begin{figure}[H]
	\centering
	\includegraphics[height = 7cm]{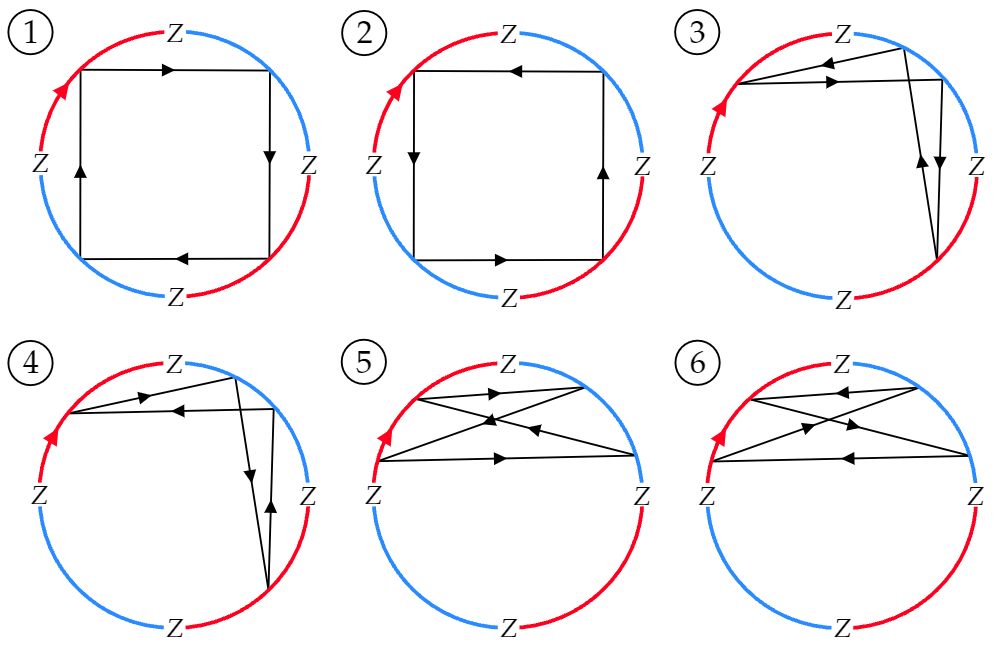}
	\caption{The six qualitatively different configurations of the four-cycle.}\label{4_cycle_list_OTOC_minus}
\end{figure}
All grey edges have been suppressed. Keeping them suppressed and using the rule for single-index edges, the perimeter ring of the diagram fragments into a number of sub-diagrams. An example of the fragmentation scheme (1) is shown below.

\begin{figure}[H]\label{example fragmentation}
	\centering
	\includegraphics[height = 3.4cm]{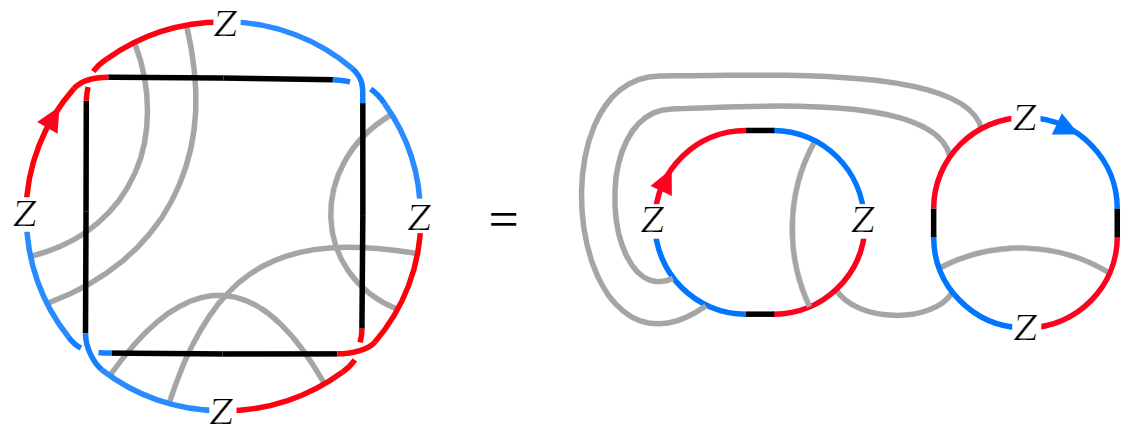}
\end{figure}
We have left a small portion of the single-index edges in the sub-diagrams on the right hand side to aid understanding of the fragmentation. From now on, we will omit this and simply connect together the blue and red contours. Suppressing the doubled-edges, we find for each fragmentation scheme, the following
\begin{figure}[H]
	\centering
	\includegraphics[height = 7cm]{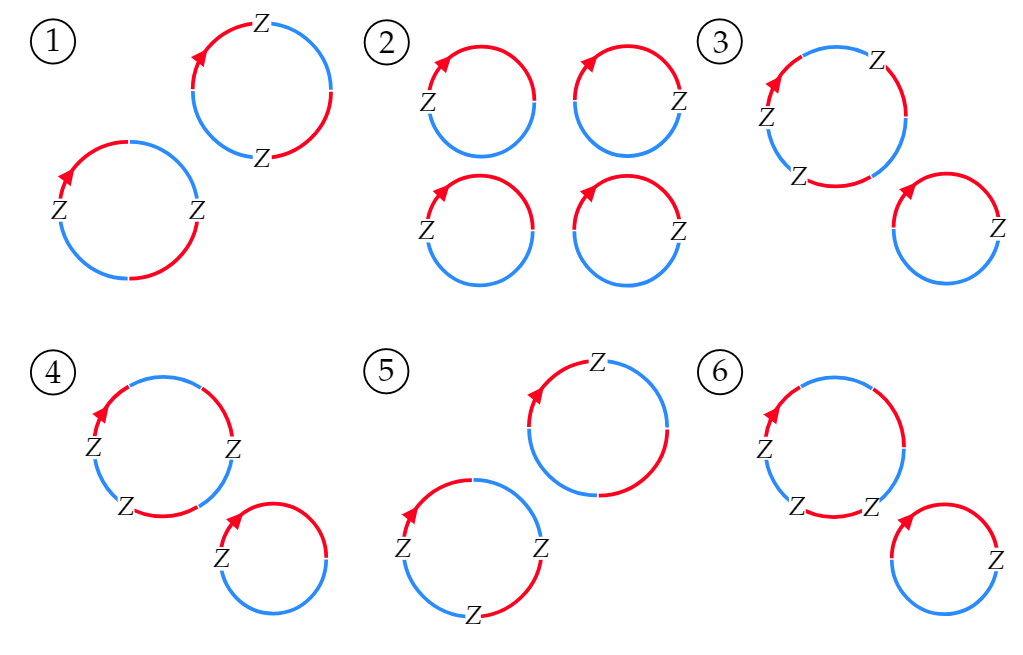}
	\caption{The four-cycle fractures the diagram into one of six qualitatively different fragmentation schemes.\label{4_cycle_list_rule_applied}}
\end{figure}
Cases (1),(3)-(6) all have $N_S=2$ and so the diagrams of interest are those with disconnected planar sub-diagrams. Consider a planar sub-diagram with only one red and one blue contour, the general form of such a diagram is shown below.
\begin{figure}[H]
	\centering
	\includegraphics[height = 2.2cm]{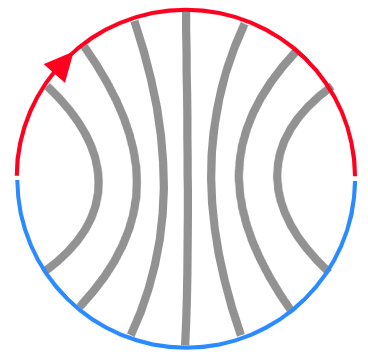}
\end{figure}
If there happens to be a $Z$ decoration at the interface between a red and blue contour, as is the case in cases (3), (4) and (5), then the trace of the $Z$ decoration appears in the evaluation of the diagram.
\begin{figure}[H]\label{once_decorated_planar}
	\centering
	\includegraphics[height = 2.5cm]{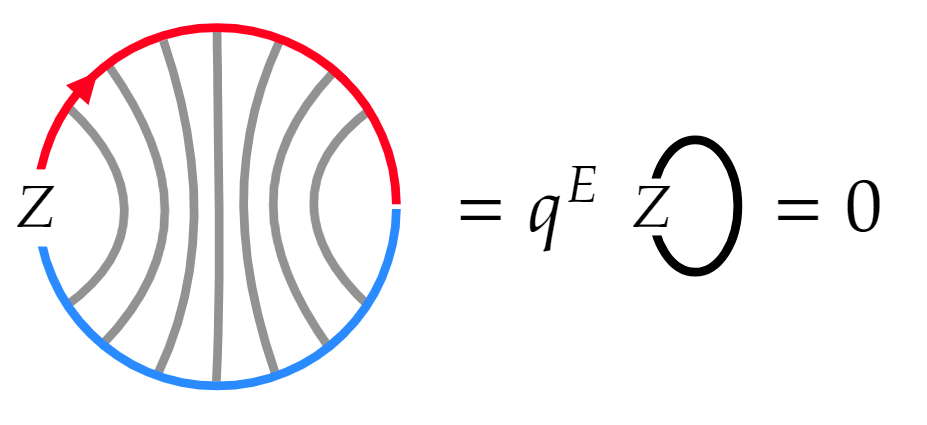}
\end{figure}
Therefore, cases (3), (4) and (6) can be discarded\footnote{If, due to the placement of a four-cycle vertex directly next door to a contour interface, one of these contours had zero length, then the sub-diagram would have only a single coloured contour. Unless this second contour also had zero length, then edges connected to this contour must connect to a different sub-diagram, this we have argued is a sub-leading contribution. If indeed both the red and blue contour were of zero length, then the $Z$ decoration must be traced and hence the contribution is zero. In any case, the conclusion about cases (3), (4) and (6) is the same.}. Planar diagrams with two red and two blue contours have a slightly more complex form.
\begin{figure}[H]
	\centering
	\includegraphics[height = 4cm]{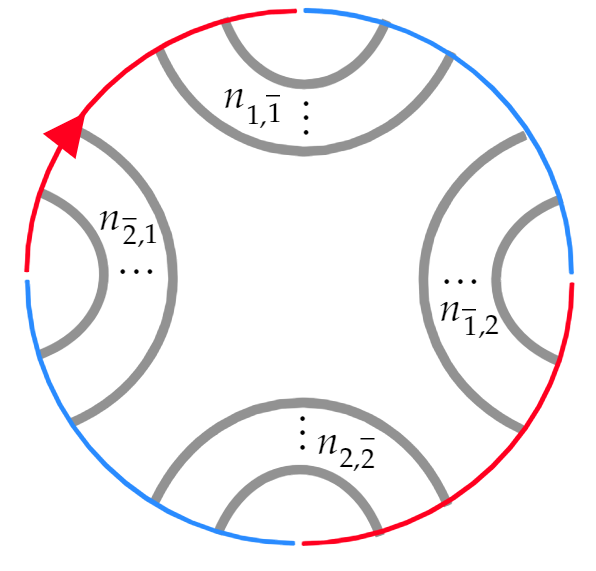}, \quad
	\includegraphics[height = 4cm]{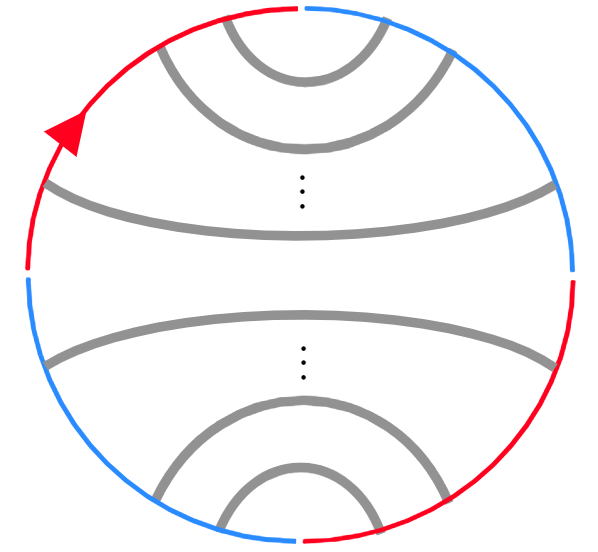}
	\caption{(Left) A general planar diagram for diagrams with two red and two blue contours. (Right) An extremal planar diagram.\label{planar 4 contour}}
\end{figure}
Where we have shown an extremal case on the right. If there happened to be a $Z$ decoration at the interface of a red and blue contour, then every non-extremal planar sub-diagram vanishes. If there happens to be an odd number of $Z$ decoration at the interfaces of red and blue contours, then even the extremal planar diagrams vanish\footnote{Addressing the concern of zero length contours once more:
If the lowest of the single-index edges in (5) of Fig. \ref{4_cycle_list_OTOC_minus} is not exactly horizontal, then each of the sub-diagrams in (5) of Fig. \ref{4_cycle_list_rule_applied} have a different number of $U$ and $U^\dagger$'s and hence cannot be planar. We are only interested in planar sub-diagrams for $N_S=2$. Assuming then that this lowest single-index edge is horizontal and brought immediately above the $Z$'s at the interface, then one of the resulting sub-diagrams in (5) of Fig. \ref{4_cycle_list_rule_applied} has only one red and blue contour and a $Z$ at one of the interfaces. We have seen that such a diagram is zero due to tracing of the $Z$ decoration. Therefore, case (5) still vanishes despite the subtlety of zero length contours.}. Therefore, case (5) can be discarded. This leaves only case (1) and (2).

Next, consider the case (2), where $N_S=4$. We must check both the case of four disconnected planar sub-diagrams, the case of disconnected sub-diagrams with all but one sub-diagram being planar and the case with two planar sub-diagrams and where the remaining two sub-diagrams reduce to the special case. Notice, however, that each sub-diagram has only one red and one blue contour and a single $Z$ decoration at one of the interfaces. Planar diagrams will force the trace to be taken of these $Z$ decorations, rendering them zero. Therefore, we can discard case (2) also\footnote{Should any of the contours be zero length, then that sub-diagram could not possible be planar.}.

We have recently seen the general form of a planar sub-diagram in Fig. \ref{planar 4 contour} with two red and two blue contours. This applies to case (1) and we see that the following extremal planar diagrams manage to pair up the $Z$ decorations onto the same index loop. Non-extremal planar diagrams would take the trace of these $Z$ decorations, as would extremal diagrams with the wrong orientation.
\begin{figure}[H]
	\centering
	\includegraphics[height = 2.5cm]{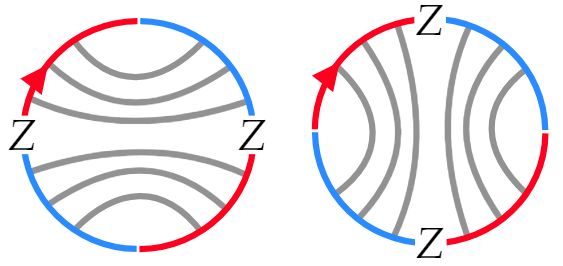}
\end{figure}

Therefore, the only contribution to $\langle \textrm{OTOC} \rangle_{-}$ is case (1) with disconnected extremal planar sub-diagrams. Reversing the fragmentation we find the following equivalent diagram.
\begin{figure}[H]\label{OTOC_minus_contributions}
	\centering
	\includegraphics[height = 5cm]{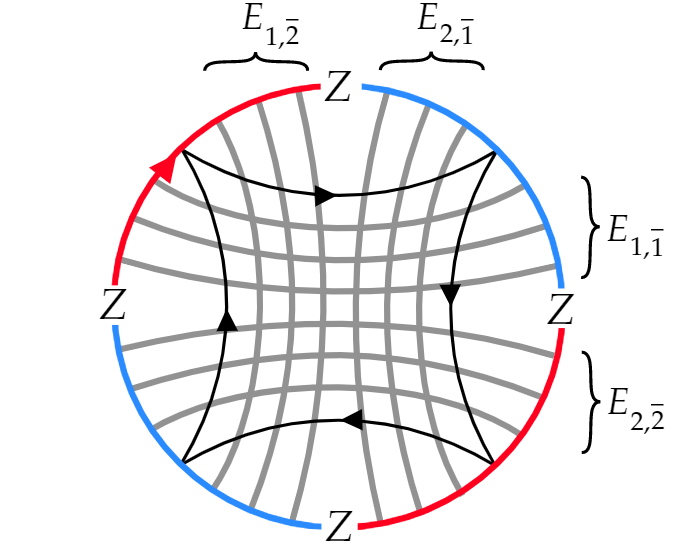}
\end{figure}
We must once again use the fact that the number of vertices in each arc is equal to argue that $E_{1,\overline{1}}=E_{2,\overline{2}}$ and $E_{1,\overline{2}}=E_{2,\overline{1}}$. Folding this diagram back into the OTOC contour gives
\begin{equation}\label{OTOCminus_with_projectors}
	\langle \textrm{OTOC} \rangle_{-} = \sum_{N_-=0}^{T-1} \frac{1}{q^2} \raisebox{-0.63\totalheight}{\includegraphics[height = 2.8cm]{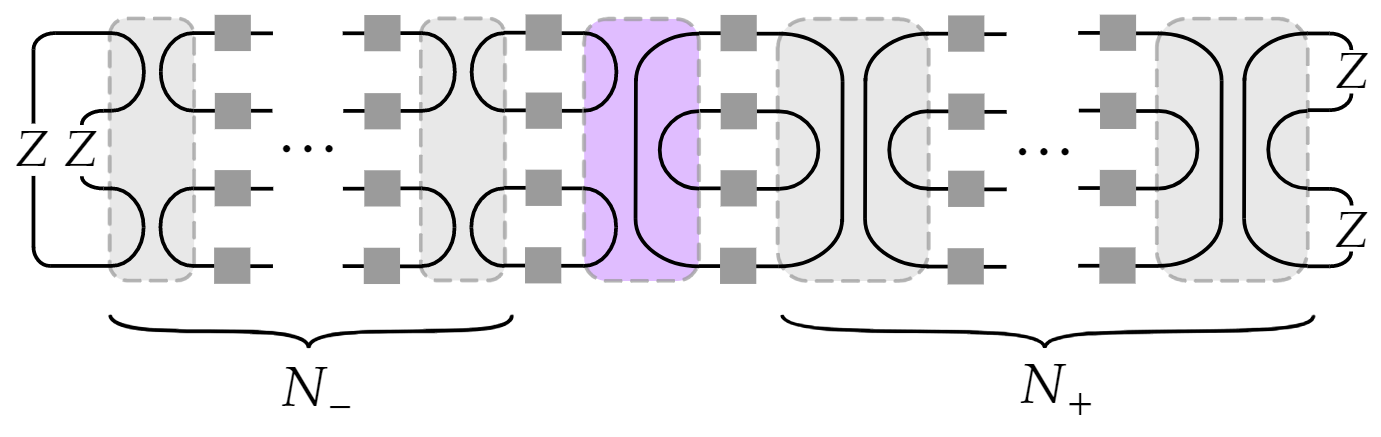}}
\end{equation}
Where $N_{+} + N_{-} +1 = T$. We associated with each closed loop a factor of $1/q$ to keep each loop normalised.

\subsection{Piecing everything together}
Recalling definition Eq. \ref{plusminusdef} and the definition of the layer $\Gamma(t)$ in Fig. \ref{physical OTOC}, $\langle \textrm{OTOC} \rangle_+$ is given by
\begin{equation}
    \langle \textrm{OTOC} \rangle_{+} = \frac{1}{q^2}\sum_{m=1}^{T-1} \left(\prod_{t=1}^{m-1}\bra{-}\Gamma(t)\ket{-}\right)q\bra{-}\Gamma(m)\ket{+}\left(\prod_{t=m+1}^{T-1}\bra{+}\Gamma(t)\ket{+}\right)
\end{equation}
and $\langle \textrm{OTOC} \rangle_-$ is given by
\begin{equation}
    \langle \textrm{OTOC} \rangle_{-} = \frac{1}{q^2}\sum_{m=0}^{T-1} \left(\prod_{t=1}^{m}\bra{-}\Gamma(t)\ket{-}\right)\left(\prod_{t=m+1}^{T-1}\bra{+}\Gamma(t)\ket{+}\right)
\end{equation}
Adding these together gives,
\begin{align}
    \int dU \textrm{OTOC} = \frac{1}{q^2}&\sum_{m=1}^{T-1} \left(\prod_{t=1}^{m-1}\bra{-}\Gamma(t)\ket{-}\right)q\bra{-}\Gamma(m)\ket{0}\left(\prod_{t=m+1}^{T-1}\bra{+}\Gamma(t)\ket{+}\right)\nonumber\\
    &- \frac{1}{q^2}\prod_{t=1}^{T-1}\bra{+}\Gamma(t)\ket{+}.
\end{align}
With a little work, which we will not do here, one can check that this is equivalent, at $\order{1/q^2}$ to
\begin{equation}
    \int dU \textrm{OTOC} = \raisebox{-0.45\totalheight}{\includegraphics[height = 1.8cm]{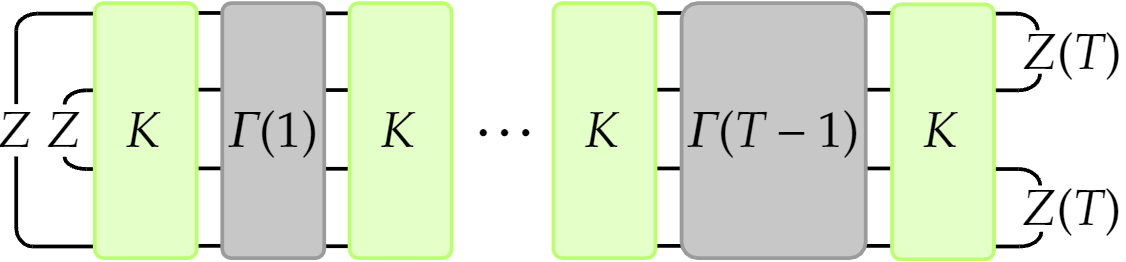}}
\end{equation}
where $K$ is as defined Eq. \ref{Kdef}.
\end{proof}
This result is in fact identical, at $\order{1/q^2}$, to the result obtained when each layer $\Gamma(t)$ is scrambled by independently random unitaries. We conjecture that when each layer shares the same scrambler, correlations between different layers in time are suppressed are exponentially suppressed in the separation distance $\Delta$, $1/q^{\Delta}$.
We did not require any information about the operators content of each layer $\Gamma(t)$ to arrive at this result and therefore, the result holds for any OTOCs of the form $\langle Z A Z(T) B^\dagger Z C Z(T) D^\dagger \rangle$
where $A$, $B$, $C$ and $D$ are time-ordered products of the form $O_1(1)\cdots O_{T-1}(T-1)$ for any normalised, traceless $O_i\in \mathbb{C}^{q\times q}$.

\section{Conclusion}
In this paper, we have investigated a selection of $n$-point correlation functions with random unitary Floquet dynamics, finding the large $q$ scaling behaviour (of the Haar average) of these correlators and their higher moments. For the second moment we have found the precise proportionality constant, given by the degree of cyclic symmetry of the correlator. We then gave an exact expression for the large $q$ behaviour of a class of OTOCs that prove important in circuit averaging problems in random Floquet circuits \cite{MMFpaper}. In doing so, we developed a diagrammatic scheme in which leading order contributions are easily identified and evaluated. Using this methodology, one could generalise this further by relaxing the restrictions on the form of OTOCs considered and by investigating correlators with different time-orderings.

\section{Acknowledgements}
I thank Curt von Keyserlingk for many useful discussions. This work is supported by an EPSRC studentship.

\bibliography{global} 

\begin{thebibliography}{10}

\bibitem{Sekino08}
Y.~Sekino and L.~Susskind, ``Fast scramblers,'' {\em Journal of High Energy
  Physics}, vol.~2008, no.~10, p.~065, 2008.

\bibitem{Lashkari2013}
N.~Lashkari, D.~Stanford, M.~Hastings, T.~Osborne, and P.~Hayden, ``Towards the
  fast scrambling conjecture,'' {\em Journal of High Energy Physics},
  vol.~2013, no.~4, p.~22, 2013.

\bibitem{Shenker2014a}
S.~H. Shenker and D.~Stanford, ``Black holes and the butterfly effect,'' {\em
  Journal of High Energy Physics}, vol.~2014, no.~3, p.~67, 2014.

\bibitem{Shenker2014b}
S.~H. Shenker and D.~Stanford, ``Multiple shocks,'' {\em Journal of High Energy
  Physics}, vol.~2014, no.~12, p.~46, 2014.

\bibitem{Shenker2015}
S.~H. Shenker and D.~Stanford, ``Stringy effects in scrambling,'' {\em Journal
  of High Energy Physics}, vol.~2015, no.~5, p.~132, 2015.

\bibitem{Maldacena2016}
J.~Maldacena, S.~H. Shenker, and D.~Stanford, ``A bound on chaos,'' {\em
  Journal of High Energy Physics}, vol.~2016, no.~8, p.~106, 2016.

\bibitem{Hartman2013}
T.~Hartman and J.~Maldacena, ``Time evolution of entanglement entropy from
  black hole interiors,'' {\em Journal of High Energy Physics}, vol.~2013,
  no.~5, p.~14, 2013.

\bibitem{Liu14a}
H.~Liu and S.~J. Suh, ``Entanglement tsunami: Universal scaling in holographic
  thermalization,'' {\em Phys. Rev. Lett.}, vol.~112, p.~011601, Jan 2014.

\bibitem{Liu14b}
H.~Liu and S.~J. Suh, ``Entanglement growth during thermalization in
  holographic systems,'' {\em Phys. Rev. D}, vol.~89, p.~066012, Mar 2014.

\bibitem{Mezei16}
M.~Mezei and D.~Stanford, ``On entanglement spreading in chaotic systems,''
  {\em Journal of High Energy Physics}, vol.~2017, p.~65, May 2017.

\bibitem{Blake16}
M.~Blake, ``Universal charge diffusion and the butterfly effect in holographic
  theories,'' {\em Physical Review Letters}, vol.~117, Aug 2016.

\bibitem{Stanford2016}
D.~Stanford, ``Many-body chaos at weak coupling,'' {\em Journal of High Energy
  Physics}, vol.~2016, no.~10, p.~9, 2016.

\bibitem{Asplund15}
C.~T. Asplund, A.~Bernamonti, F.~Galli, and T.~Hartman, ``Entanglement
  scrambling in 2d conformal field theory,'' {\em Journal of High Energy
  Physics}, vol.~2015, Sep 2015.

\bibitem{Banerjee2017}
S.~Banerjee and E.~Altman, ``Solvable model for a dynamical quantum phase
  transition from fast to slow scrambling,'' {\em Physical Review B}, vol.~95,
  Apr 2017.

\bibitem{Roberts18}
D.~A. Roberts, D.~Stanford, and A.~Streicher, ``Operator growth in the syk
  model,'' {\em Journal of High Energy Physics}, vol.~2018, Jun 2018.

\bibitem{Roberts16}
D.~A. Roberts and B.~Swingle, ``Lieb-robinson bound and the butterfly effect in
  quantum field theories,'' {\em Phys. Rev. Lett.}, vol.~117, p.~091602, Aug
  2016.

\bibitem{Swingle17}
D.~{Chowdhury} and B.~{Swingle}, ``{Onset of many-body chaos in the $O(N)$
  model},'' {\em ArXiv e-prints}, Mar. 2017.

\bibitem{Aleiner16}
I.~L. Aleiner, L.~Faoro, and L.~B. Ioffe, ``Microscopic model of quantum
  butterfly effect: Out-of-time-order correlators and traveling combustion
  waves,'' {\em Annals of Physics}, vol.~375, pp.~378 -- 406, 2016.

\bibitem{CalabreseCardy05}
P.~Calabrese and J.~Cardy, ``Evolution of entanglement entropy in
  one-dimensional systems,'' {\em Journal of Statistical Mechanics: Theory and
  Experiment}, vol.~2005, no.~04, p.~P04010, 2005.

\bibitem{Nahum16}
A.~Nahum, J.~Ruhman, S.~Vijay, and J.~Haah, ``Quantum entanglement growth under
  random unitary dynamics,'' {\em Phys. Rev. X}, vol.~7, p.~031016, Jul 2017.

\bibitem{Nahum17}
A.~Nahum, S.~Vijay, and J.~Haah, ``Operator spreading in random unitary
  circuits,'' {\em Phys. Rev. X}, vol.~8, p.~021014, Apr 2018.

\bibitem{RvK17}
C.~W. von Keyserlingk, T.~Rakovszky, F.~Pollmann, and S.~L. Sondhi, ``Operator
  hydrodynamics, otocs, and entanglement growth in systems without conservation
  laws,'' {\em Phys. Rev. X}, vol.~8, p.~021013, Apr 2018.

\bibitem{OTOCDiff1}
V.~Khemani, A.~Vishwanath, and D.~A. Huse, ``Operator spreading and the
  emergence of dissipative hydrodynamics under unitary evolution with
  conservation laws,'' {\em Phys. Rev. X}, vol.~8, p.~031057, Sep 2018.

\bibitem{OTOCDiff2}
T.~Rakovszky, F.~Pollmann, and C.~W. von Keyserlingk, ``Diffusive hydrodynamics
  of out-of-time-ordered correlators with charge conservation,'' {\em Phys.
  Rev. X}, vol.~8, p.~031058, Sep 2018.

\bibitem{Brown12}
W.~{Brown} and O.~{Fawzi}, ``{Scrambling speed of random quantum circuits},''
  {\em ArXiv e-prints}, Oct. 2012.

\bibitem{ChanDeLuca1}
A.~Chan, A.~De~Luca, and J.~T. Chalker, ``Solution of a minimal model for
  many-body quantum chaos,'' {\em Phys. Rev. X}, vol.~8, p.~041019, Nov 2018.

\bibitem{Hayden07}
P.~Hayden and J.~Preskill, ``Black holes as mirrors: quantum information in
  random subsystems,'' {\em Journal of High Energy Physics}, vol.~2007, no.~09,
  p.~120, 2007.

\bibitem{Cotler2017a}
J.~Cotler, N.~Hunter-Jones, J.~Liu, and B.~Yoshida, ``Chaos, complexity, and
  random matrices,'' {\em Journal of High Energy Physics}, vol.~2017, Nov 2017.

\bibitem{Gharibyan2018}
H.~Gharibyan, M.~Hanada, S.~H. Shenker, and M.~Tezuka, ``Onset of random matrix
  behavior in scrambling systems,'' {\em Journal of High Energy Physics},
  vol.~2018, Jul 2018.

\bibitem{Meh2004}
M.~L. Mehta, {\em Random Matrices}.
\newblock Elsevier, 3rd~ed., 2004.

\bibitem{Brezin1997}
E.~Brézin and S.~Hikami, ``Spectral form factor in a random matrix theory,''
  {\em Physical Review E}, vol.~55, p.~4067–4083, Apr 1997.

\bibitem{RobertsDesign}
D.~A. Roberts and B.~Yoshida, ``Chaos and complexity by design,'' {\em Journal
  of High Energy Physics}, vol.~2017, p.~121, Apr 2017.

\bibitem{MMFpaper}
E.~R. McCulloch and C.~W. von Keyserlingk, ``Operator spreading in the memory
  matrix formalism.'' [Unpublished].

\bibitem{Collins2002MomentsAC}
B.~Collins, ``Moments and cumulants of polynomial random variables on
  unitarygroups, the itzykson-zuber integral, and free probability,'' {\em
  International Mathematics Research Notices}, vol.~2003, pp.~953--982, 2002.

\bibitem{Collins2006}
B.~Collins and P.~Śniady, ``Integration with respect to the haar measure on
  unitary, orthogonal and symplectic group,'' {\em Communications in
  Mathematical Physics}, vol.~264, p.~773–795, Mar 2006.

\bibitem{asymptoticweingarten}
D.~Weingarten, ``Asymptotic behavior of group integrals in the limit of
  infinite rank,'' {\em Journal of Mathematical Physics}, vol.~19, no.~5,
  pp.~999--1001, 1978.

\bibitem{Diaconis_2001}
P.~Diaconis and S.~N. Evans, ``Linear functionals of eigenvalues of random
  matrices,'' {\em Transactions of the American Mathematical Society},
  vol.~353, no.~7, pp.~2615--2633, 2001.

\end{thebibliography}
\bibliographystyle{ieeetr}

\end{document}